\documentclass{article}
\pdfoutput=1
\usepackage[T1]{fontenc}
\usepackage[linktocpage=true]{hyperref}

\usepackage[totalwidth=18cm,totalheight=24cm]{geometry} 

\usepackage[english]{babel}	
\usepackage[protrusion=true,expansion=true]{microtype}	
\usepackage{mathtools,mathrsfs,amsmath,amsfonts,amssymb,amsthm,bbm}

\usepackage{tikz-cd} 
\DeclareFontFamily{OT1}{pzc}{}
\DeclareFontShape{OT1}{pzc}{m}{it}{<-> s * [1.10] pzcmi7t}{}
\DeclareMathAlphabet{\mathpzc}{OT1}{pzc}{m}{it}
\usepackage{cases}
\usepackage{enumerate}
\usepackage{wrapfig}
\usepackage{url}

\usepackage{indentfirst} 
\usepackage{cite}

\renewenvironment{thebibliography}[1]{
  \begin{oldthebibliography}{#1}
    \setlength{\itemsep}{0.65em}
    \setlength{\parskip}{0em}
}
{
  \end{oldthebibliography}
}

\usepackage{sectsty}
\allsectionsfont{\centering}

\usepackage{fancyhdr}
\numberwithin{equation}{section}		
\numberwithin{figure}{section}			
\numberwithin{table}{section}				

\newcommand*{\al}{\ensuremath{\alpha}}
\newcommand*{\be}{\ensuremath{\beta}}
\newcommand*{\ga}{\ensuremath{\gamma}}
\newcommand*{\Ga}{\ensuremath{\Gamma}}
\newcommand*{\de}{\ensuremath{\delta}}

\newcommand*{\eps}{\ensuremath{\epsilon}}

\newcommand*{\la}{\ensuremath{\lambda}}
\newcommand*{\La}{\ensuremath{\Lambda}}

\newcommand*{\Si}{\ensuremath{\Sigma}}

\newcommand*{\om}{\ensuremath{\omega}}

\newcommand*{\pa}{\partial}
\newcommand*{\bra}{\langle}
\newcommand*{\ket}{\rangle}
\newcommand*{\ra}{\rightarrow}
\newcommand*{\rt}{\triangleright}
\newcommand*{\tr}{\mathrm{tr}}
\newcommand*{\Tr}{\mathrm{Tr}}

\newcommand*{\mr}{\mathrm} 
\newcommand*{\mc}{\mathcal} 
\newcommand*{\ms}{\mathscr} 

\newtheorem*{lemma}{Lemma}

\title{\vfill
\textbf{Poincar\'{e}-Pleba\'{n}ski formulation of GR \\
		 and dual simplicity constraints}
}

\author{Vadim Belov~\thanks{E-mail: vadim.belov@desy.de} \vspace{5pt}
\\
\emph{\normalsize II. Institute for Theoretical Physics, University of Hamburg,} \\ 
\emph{\normalsize Luruper Chaussee 149,} \\
\emph{\normalsize 22761 Hamburg, Germany}
}

\date{\normalsize \today}

\begin{document}

\maketitle
\thispagestyle{empty}

\addto{\captionsenglish}{\renewcommand{\abstractname}{Abstract}}
\begin{abstract}
We revise the classical continuum formulation behind the Spin Foam approach to the quantization of gravity. Based on the recent applications of the current EPRL-FK model beyond triangulations, we identify the tension with the implementation of the `volume' part of simplicity constraints, required for the passage from the topological BF theory to gravity. The crucial role, played by 4d normals in the linear version of constraints, necessitates the extension of the configuration space, and we argue to switch from normal 3-forms directly to tetrads. The requirement of vanishing torsion leads to consider first an unconstrained extended Poincar\'{e} BF theory, which we characterize fully both at the Lagrangian and Hamiltonian levels, paying special attention to its gauge symmetries. The simplicity constraints are introduced naturally, in the spirit of Pleba\'{n}ski formulation, and we give their tetradic version, dual to that of using 3-forms. This brings us much closer to the geometric content of General Relativity.
\end{abstract}

\hfill

\tableofcontents
\vfill
\newpage
\section{Introduction and motivation}

Loop Quantum Gravity (LQG) and Spin Foams (SF) are the two ambitious non-perturbative approaches to the quantization of gravity. They are usually thought of to be complementary to each other, although not entirely equivalent. The precise correspondence between them is still an open issue. The distinction shows up already at the level of the classical starting points of the two theories, operating with the two dissimilar sets of variables. The first one performs the canonical quantization of the tetradic Holst action~\cite{Holst1996action} (in the time gauge). Whilst the second one is the implementation of the discrete path-integral for gravity~\cite{Perez2013SF-review}, based on its Pleba\'{n}ski reformulation as a constrained BF theory. The latter naturally generalizes the notion of the state-sum/partition function of Topological Quantum Field Theory (TQFT)~\cite{Barrett1995QG-TQFT,Oeckl2003general-bndry}. It can be used as a way to define the dynamics of LQG, by constructing (the family of) transition amplitudes (for the kinematical states of canonically quantized theory). 



In Spin Foams, the discretized variables of BF theory live on (abstract) 2-complexes, typically thought of as being dual to triangulations. The relation to General Relativity (GR), in the form of some discrete geometry data, is established through imposition of the simplicity constraints. The current EPRL~\cite{EPRL-FK2008flipped2,EPRL-FK2008finiteImmirzi} and FK~\cite{EPRL-FK2008FK,EPRL-FK2008LS} models are based on the so-called linear formulation, involving the time-normal vector field. In result, the boundary states are spanned by the `projected spin-networks', where these normals, discretized at the nodes, appear as arguments on par with connection variables~\footnote{The role of these normals has been highlighted by S. Alexandrov on the grounds of Lorentz covariance, see \cite{AlexandrovRoche2011CovariantLQG-critique} and references therein; they have also been studied in the context of (extended) Group Field Theory formalism with non-commuting variables~\cite{BaratinOriti2012BO-model}}. The EPRL quantization map thus defines the parameters of the covariant lift of the usual LQG spin-network states~\cite{LivineDupuis2010lifting-SU2}, and one may regard each SF as contribution to the sum over histories of canonically quantized geometries~\cite{ReisenbergerRovelli1997SFfromLQG-sum}. The present work reconsiders the classical formulation of linear simplicity constraints~\cite{GielenOriti2010Plebanski-linear}, and, in particular, the role of the degrees of freedom (d.o.f.) behind the 4d normals, based on the recent findings in the SF asymptotics.


In the absence of experiment guidance, the major consistency check for the model is the peakedness on the classical geometries, namely on the solutions of the Einstein equations, in a suitable semi-classical regime ($\hbar\ra0$). Using Barrett's reconstruction theorem and (extended) stationary phase methods (when the typical scale -- the physical area in Planck units -- is large), the EPRL-FK amplitude correctly reproduces the 4-simplex geometry and Regge gravity, for certain fixed, non-degenerate boundary data~\cite{Barrett-etal2009asympt4simplex-euclid,Barrett-etal2010asympt4simplex-lorentz}. On the other hand, the KKL extension~\cite{KKL2010AllLQG,KKL2010correctedEPRL} of the simplicial EPRL amplitude to the graphs of arbitrary valence, has not been thoroughly studied until very recently~\cite{Dona-etal2017KKL-asumpt}, although the partial results were available in the symmetry reduced setting~\cite{BahrSteinhaus2016Cuboidal-EPRL}. These results include appearance of certain `non-geometric' configurations, demonstrating shape-mismatch, as well as non-zero physical norm of states with torsion. They are SF analogues of LQG's `twisted' geometries~\cite{FreidelSpeziale2010Twisted-geometries} (discontinuous over flat faces), or torsionful  `spinning' (continuous over arbitrarily curved faces) piecewise-flat geometries~\cite{FreidelGeillerZiprick2013cont-LQG-phase-space,FreidelZiprick2014Spinning-geometry}~\footnote{The non-zero torsion generically presents in LQG phase space by the Lemma 2 in~\cite{FreidelGeillerZiprick2013cont-LQG-phase-space}, since the Ashtekar-Barbero connection mixes up extrinsic curvature $A=\Ga[e]+\ga \mathrm{K}$ (its contribution is governed by the Immirzi parameter), residing over the edges of the cellular complex.}; whilst Regge configurations appear only as a constrained subset~\cite{DittrichRyan2011simplicial-phase-space}.

In Sec.~\ref{ssec:volume-fate}, we scrutinize the example of hypercuboid and trace back the `non-geometricity' to the way how the simplicity constraints are imposed in the classical theory. In particular, the possibility to neglect the 4-volume constraint in the simplex does not hold for more complicated polytopes, so that Barrett's reconstruction is not applicable. We then proceed with the application to the same system of the fully linear treatment due to Gielen and Oriti~\cite{GielenOriti2010Plebanski-linear}, with independent normals. The workings of their `linear volume' constraint prompt to switch from the normals (3-forms) directly to edge lengths (tetrads) as new independent variables, using the Hodge duality. In the rest of the paper we study the implications of that change on the classical continuum theory.

First, one may want to incorporate the 4d closure condition on normals (e.g. in the 4-simplex~\cite{GielenOriti2010Plebanski-linear}), or -- as we show to be closely related -- the vanishing torsion (i.e. 2d closure for tetrads). In addition to already noted possible link between `non-gemetricity' and torsion, it has also been argued on the basis of more involved examples, such as $n$-point correlation functions and extended triangulations with the bulk curvature. In the latter example, the actual details of taking (semi-)classical limit require to invoke the flipped regime of small Barbero-Immirzi parameter $\ga\ra 0$, in addition to the usual scaling of spins $j\ra\infty$. Supposedly, the large $\ga^{-1}$ in front of the Holst term in the exponent weight of a path integral leads to the dominant contributions from the stationary phase/critical configurations, satisfying the respective equations of motion (e.o.m.) -- namely, the simplicial version of the Cartan structure equation: the rotations are trivial (up to~$\pi$) in the plane of the hinge. Thus the Immirzi parameter effectively takes control over the strength of another ``geometricity constraints ... needed in order to reduce the SF dynamical variables to the configurations compatible with the metric geometry of the triangulation'' (see~\cite{MagliaroPerini2013asympt-Regge} and references therein).

However, neither 4d closure, nor zero-torsion is the part of the theory in the same sense as the bivector closure. Namely, there is no corresponding dynamical law (like the Gauss constraint) from which such condition would follow as (discretized) equations of motion. Therefore we propose in the Sec.~\ref{sec:Poincare-BF} the modification of the BF action by the zero-torsion, imposed via Lagrange multipliers. This also seems natural from the point of view of the contact with GR, rather then the Einstein-Cartan (EC) theory. We perform a comprehensive study of this model system prior to imposition of simplicity constraints. This turns out to be topological, with the gauge group being the non-homogeneous Poincar\'{e} (affine) extension of the usual homogeneous Lorentz group~\footnote{Such a model was first considered in~\cite{Bi2013PoincareBF}, as pointed out by the referee. It is closely related (equivalent via integration by parts, i.e. up to boundary terms) to the corresponding topological Higher Gauge Theory, or 2-group Categorical Generalization~\cite{GirelliPfeifferPopescu2008BFCG}.}. The gauge symmetries are defined generically as leaving the action invariant off-shell (up to divergence), and we explicitly derive them from the form of equations of motion in the covariant framework, using the converse of the Noether's 2nd theorem. The Dirac's constrained Hamiltonian analysis is performed, having the aim to demonstrate via explicit construction of the gauge generator, which maps solutions of e.o.m. onto solutions, that the full 4d symmetry persists on the canonical level as well, although the manifest covariance may be explicitly broken. 

At last, the manifest presence of the tetrad frame field $e$ in the formalism among configuration variables -- instead of the normal 3-forms -- makes the introduction of simplicity constraints especially natural~\footnote{This stratedy was also advocated within the constrained BFCG approach~\cite{Mikovic-etal2012BFCG-Poincare}, from which our differs in 2 respects: 1) no categorical generalization is implied or required, cf.~\cite{GegenbergMann1999Poincare-BF-gravity}, 2) exploration of the new form of \emph{linear} simplicity constraints.}. We furthermore propose to look for the linear formulation of~\cite{GielenOriti2010Plebanski-linear} in the new guise and introduce its dual version in Sec.~\ref{sec:Poincare-Plebanski}. The equivalence with the usual simplicity of bivectors is proven. The advantage is in the clear-cut geometric interpretation of the `volume' part of simplicity constraints in terms of an actual 3-volume, in analogy to the 4-volume of the quadratic case.
Besides, the separation between Pleba\'{n}ski and 1st order formulations gets blurred to some extent, as their variables are brought together. This resonates nicely with our very first comment on the dissimilarity of the two classical foundations of SF and LQG, respectively. 

The structure is as follows. The first half of the paper consists of the brief recap on quadratic Pleba\'{n}ski formulation (including the role of 4-volume) in Sec.~\ref{sec:classics}, and our revision of the constraints in the EPRL-FK-type models in Sec.~\ref{sec:quantization}. Special attention is paid to the role of normals. The familiar reader may skip the exposition and start reading with Sec.~\ref{ssec:volume-fate}, where the case study of the volume constraints in the EPRL-KKL hypercuboid set-up is performed. This should justify our shift from normals to tetrads in the classical study of the second part of the paper. The Sec.~\ref{sec:Poincare-BF} presents the self-contained primer on the Poincar\'{e}-BF theory and, as such, can be read independently. The Sec.~\ref{sec:Poincare-Plebanski} contains reformulation of linear simplicity in terms of new variables. Finally, we comment on the relations between the various action principles, draw some conclusions in Sec.~\ref{sec:summary} and discuss on possible outlook for quantization.


\section{The recap on constraints: classics}\label{sec:classics}
\paragraph{The setting.} The classical backdrop behind the Spin Foam quantization program  is the observation due to Pleba\'{n}ski~\cite{Plebanski1977,Celada-etal2016BF-review} that the Einstein-Cartan action can be recast as a constrained BF theory:
\begin{equation}\label{eq:BF+cnstr}
S[B,\om, \la] \ = \ \int_{\mc M} B_{AB}^{\phantom{AB}}\wedge F^{AB}[\om] \ + \ \la^\al  \mc C_\al[B].
\end{equation}
Here $A,B=0,1,2,3$ -- (internal, or anholonomic) indices in the defining vector representation of the homogeneous Lorentz group $H=SO(3,1)$; $B$ is the Lie algebra $\mathfrak{h=so}(3,1)$-valued 2-form (or $\mathfrak{h=so}(4)$ for Euclidean signature spacetime~$\mc M$), transforming in the (co)adjoint representation; $\om$~is the spin-connection with the curvature $F$. The 1st $\int BF$ term, taken on its own, defines a topological field theory without local degrees of freedom; it admits a well-defined exact state sum quantization, discretized over 2-complexes (\textit{\`{a} la} Spin Foam). The 2nd term represents constraints $\mc C_\al[B]=0$ on $B$-field, enforced by the Lagrange multipliers $\la^\al$ ($\al$ -- multi-index); they effectively reduce the number of independent $B$-components, so that the bivector is given by the (dual) simple product $B=\star e\wedge e$ of (some) tetrad frame field $e$. Hence, on the constraint surface, the theory acquires the form due to Einstein-Cartan:
\begin{equation}\label{eq:Einstein-Cartan}
S_{\mathrm{EC}}[e,\om] \ = \ \int_{\mc M} \frac12 \eps_{ABCD}^{\phantom{ABCD}}\,e^A\wedge e^B\wedge F^{CD}[\om].
\end{equation}
In its turn, it is widely accepted as a 1st order formulation of a theory of gravity, since given the equations of motion for $\om$ in vacuum are satisfied, this renders the theory (on-shell) to the 2nd order tetradic Einstein-Hilbert action: $S_{\mathrm{EC}}\big|_{\de\om} = S_{\mathrm{EC}}[e,\om[e]] \equiv S_{\mathrm{EH}}[e]$.

\paragraph{The strategy} in the majority of Spin Foam approaches is to \textit{first quantize and then constrain}, according to the following route:
\begin{enumerate}
\item discretize the classical theory on a piecewise-flat partition of the spacetime $\mc M$ (most commonly, simplicial);
\item quantize the topological BF part of the discretized theory;
\item impose (a version of) simplicity constraints $\mc C_\al[B]\approx 0$ directly at the quantum level.
\end{enumerate}
The non-trivial part in constructing SF models for gravity comes from the third step.

The most widely known and well studied is the Pleba\'{n}ski's quadratic set of constraints, existing in 2 versions:
\begin{equation}\label{eq:Plebanski-cnstr}
\text{(a)} \qquad B^{AB}\wedge B^{CD} \ = \ V \, \eps^{ABCD} \quad\qquad \stackrel{\tilde{V}\neq 0}{\Longleftrightarrow} \qquad\quad \text{(b)} \qquad \eps_{ABCD}^{\phantom{ABCD}}B^{AB}_{ab} B^{CD}_{cd} \ = \ \tilde{V} \, \eps_{abcd}^{\phantom{ABCD}}.
\end{equation}
They are equivalent, provided the quantity $V$ (resp. $\tilde{V}$) -- which is defined by~\eqref{eq:Plebanski-cnstr} through contraction with $\eps$ -- is non-vanishing~\cite{DePietriFreidel1999Plebanski}. In that case, there are two non-degenerate sectors of solutions:
\begin{equation}\label{eq:Plebanski-sectors}
I^\pm: \quad B^{AB} \ = \ \pm e^A\wedge e^B,  \qquad\qquad II^\pm: \quad B^{AB} \ = \ \pm \frac12 \eps^{AB}_{\phantom{AB}CD}e^C\wedge e^D,
\end{equation}
and $V=\pm\frac{1}{4!}\eps^{\phantom{A}}_{ABCD}e^A\wedge e^B\wedge e^C\wedge e^D=\tilde{V}\,d^4x$ acquires an interpretation of spacetime 4-volume.
The sectors $II^\pm$ reproduce \eqref{eq:Einstein-Cartan} up to the discrete sign ambiguity, while $I^\pm$-sectors give the topological Holst term. The treatment of degenerate case $\tilde{V}=0$, and relations between sectors can be found in \cite{Reisenberger1999Plebanski}.

\paragraph{The discretizations} of classically equivalent forms of constraints \eqref{eq:Plebanski-cnstr} lead to two, \textit{a priori} different, SF models. The (a)-case gives the version of the Reisenberger state-sum model \cite{Reisenberger1997state-sum} (corresponding to a self-dual formulation), whereas the case (b) is the most prevailing and leads to the Barrett-Crane (BC)~\cite{BarrettCrane1998Euclid,BarrettCrane2000Lorentz} and the new models. The discrete connection is captured by the finite holonomies
\begin{equation}\label{eq:discrete-holonomies}
h_e[\om] \ = \ \overrightarrow{\exp} \left(\int_e \om^{AB}\mc J_{AB}\right),
\end{equation}
path-ordered along the dual edges $e$. According to the second choice, one associates to the $B$-field the bivectors~\footnote{Mention that, when the local flatness is not assumed, the integrand should be acted upon by holonomies, referring it to the single source frame, in order to ensure the correct transformation properties of the Lie algebra element under the gauge rotations (cf.~\cite{FreidelGeillerZiprick2013cont-LQG-phase-space}).}
\begin{equation}\label{eq:discrete-B}
\bigwedge^2\mathbb{R}^{3,1} \ \ni \ B^{AB}_f \ = \ \int_{S_f} B^{AB},
\end{equation}
by integrating over the co-dimension 2 cells $S_f$ of the piecewise-flat complex, which we label bijectively with the faces $f$ of the dual 2-skeleton. Together, they constitute the discretized set of (kinematical) variables of BF theory and Pleba\'{n}ski formulation of gravity. In the latter case, the simplicity constraints should be also discretized. 

Suppose, our 2-complex is dual to a triangulation. Then, depending on the relative position of triangles in a 4-simplex, the constraints fall into 3 types:
\begin{enumerate}[(i)]
\item $\displaystyle{\eps_{ABCD}^{\phantom{ABCD}}B^{AB}_fB^{CD}_f=0}$ for each triangle/face $f$ -- diagonal (or `face') simplicity; \label{diag-simplicity}
\item $\displaystyle{\eps_{ABCD}^{\phantom{ABCD}}B^{AB}_fB^{CD}_{f'}=0}$ if two faces share an edge $f\cap f'=e$ -- cross-simplicity (or `tetrahedral' constraint); \label{cross-simplicity}
\item $\displaystyle{\eps_{ABCD}^{\phantom{ABCD}}B^{AB}_fB^{CD}_{f'}=:\tilde{V}_v(f,f')}$ for any pair of faces $f,f'$ meeting at the vertex~$v$ and spanning 4-simplex volume -- the so-called volume (or `4-simplex') constraint.\label{volume-quadr}
\end{enumerate}
Each of these constraints have different status and are treated accordingly. In particular, they are implemented, respectively, at the level of faces/tetrahedra/4-simplices.

\paragraph{The closure condition.} In addition, there is usually imposed also the 3d closure
\begin{equation}\label{eq:closure-3d}
\sum_{f\supset e} B^{AB}_f \ = \ 0 \qquad \forall \, e.
\end{equation}
This is the consequence of the BF e.o.m. $\nabla^{(\om)}_{[c}B^{AB}_{ab]}=0$ in the discrete setting: integrate over the 3d volume of the flat tetrahedron $\tau_e$, putting connection to zero via the gauge transform, and use the Stokes' theorem. It reflects the gauge invariance of the BF theory and gravity. In the canonical picture, this corresponds to the Gauss law constraint (after the symplectic reduction by the cellular flatness constraint~\cite{FreidelGeillerZiprick2013cont-LQG-phase-space}, restraining local curvature on hinges), and generates the local gauge rotations. Accordingly, in the quantum theory it is usually implemented via group integration, projecting on an invariant subspace.

\paragraph{The geometric meaning.} The holonomies give the parallel transport of tensors and spinors, taking into account the relative rotation of reference frames between the path endpoints. Regarding the bivector \eqref{eq:discrete-B}, when it comes from the metric structure (i.e. the co-tetrad $e$-field, appropriately discretized), then its norm gives the area of the corresponding triangle and the tensor structure encodes the directions of surface, in locally inertial frame of reference. Strictly speaking, such $B$ should not be considered as a variable corresponding to elementary excitations, but rather has a composite nature. The simplicity constraints formulate the necessary and sufficient conditions for the system of bivectors to correspond to the faces of the discrete cell-complex (metric; in our case, of the single 4-simplex).

The 1st condition implies that the bivector is \textit{simple}, i.e. given by the wedge product of two vectors:
\begin{equation*}
\text{\eqref{diag-simplicity}} \qquad \Rightarrow \qquad B^{AB}_{f_1} \ = \  E^{[A}_2E^{B]}_3 \qquad \text{or} \qquad \star B^{AB}_{f_1} \ = \  E^{[A}_2E^{B]}_3
\end{equation*}

If two triangles share a common edge, then the sum of the corresponding two bivectors is also simple:
\begin{equation*}
\text{\eqref{cross-simplicity}} \qquad \Rightarrow \qquad B^{AB}_{f_2} \ = \  E^{[A}_3E^{B]}_1 \qquad \text{or} \qquad \star B^{AB}_{f_2} \ = \  E^{[A}_3E^{B]}_1, \qquad \text{and cyclically }\ \forall f\subset \tau_e.
\end{equation*}

The condition \eqref{eq:closure-3d} states that the geometry of the tetrahedron $\tau_e$, built on vectors $E_1,E_2,E_3$ (or its dual), has a closed boundary. This condition allows a generalization to arbitrary valence and is sufficient to uniquely specify the geometry of a flat polyhedron~\cite{BianchiDonaSpeziale2011Polyhedra}. 

An arbitrary set of ten bivectors, satisfying the above three conditions (supplemented with the orientation reversion $B_{AB}=-B_{BA}$ + some non-degeneracy requirements) forms the so-called \textit{bivector geometry}. The utility of the concept is that it allows to \textit{reconstruct} the unique flat 4-simplex (up to the orientation, translations and inversions), as is shown in~\cite{BarrettCrane1998Euclid}. This is the geometrical underpinning behind the construction of the BC model. The role of \eqref{volume-quadr} is to ensure that the geometries of the tetrahedra fit together to form consistently a 4d geometry, in particular, that the volume of a 4-simplex is invariably defined. The volume constraint~\eqref{volume-quadr} is not the part of conditions, defining the bivector geometry, because it is implied by the constraints on the tetrahedral level and the closure. The derivation goes as follows~\cite{EPRL-FK2008flipped2,EPRL-FK2008LS}. Label the five tetrahedra with $e=1,...,5$; the triangle $\triangle_{12}$ is shared by two respective tetrahedra. Using the closure \eqref{eq:closure-3d}, say for tetrahedron $1$, and contracting it with all the other bivectors, one can freely swap between triangles, for instance:
\begin{equation}\label{eq:volume-simplex}
\eps\, B(\triangle_{12})\cdot B(\triangle_{45})+\eps\, B(\triangle_{13})\cdot B(\triangle_{45})\ = \ -\eps\, B(\triangle_{14})\cdot B(\triangle_{45})-\eps \, B(\triangle_{15})\cdot B(\triangle_{45})\ = \ 0,
\end{equation}
so that the r.h.s. eliminates on the surface of the simplicity constraints \eqref{cross-simplicity} $\Rightarrow$ hence \eqref{volume-quadr} follows.

In the canonical picture parlance, \eqref{volume-quadr} is interpreted as  a ``secondary'' constraint, which ensures the dynamical conservation of the simplicity constraints \eqref{cross-simplicity} across the 4-simplex. (This is, however, not the statement of the Hamiltonian analysis of the underlying action \cite{Buffenoir-etal2004Hamiltonian-Plebanski} in the Bergmann's terminology.) Replacement of~\eqref{volume-quadr} by~\eqref{eq:closure-3d} is particularly beneficial for the quantum theory, since the linear in $B$ and local in each tetrahedron closure constraint is much more easier to deal with. By introducing auxiliary normals to tetrahedra, one could incorporate~\eqref{diag-simplicity} and \eqref{cross-simplicity} into the single `linear cross-simplicity' constraint, retaining the same geometric picture, which led to the new EPRL-FK models~\cite{EPRL-FK2008flipped2,EPRL-FK2008finiteImmirzi,EPRL-FK2008FK,EPRL-FK2008LS}. We now discuss briefly some details of this construction, as they appear in the literature.

\section{On the quantization in new models}\label{sec:quantization}
There are various ways to arrive at SF partition function (associated with the 2-complex $\Upsilon$)
\begin{equation}\label{eq:partition-Z}
Z_{\Upsilon} \ = \ \sum_{j_f,\iota_e}\prod_{f}\mc A_f\prod_{e}\mc A_e\prod_{v}\mc A_v
\end{equation}
from the classical input laid out above. Roughly they could be captured in two types:
\begin{itemize}
\item Relying on the factorization of the representation \eqref{eq:partition-Z}, it is sufficient to \textit{quantize the geometry of a 4-simplex} and then to glue such several contributions together. In particular, this route was pursued in the original derivation of the Barrett-Crane (BC) model \cite{BarrettCrane1998Euclid}.

This may be very illuminating in determining the (kinematical) state space of the model. The vertex amplitude determines the graph's local dynamics and in the canonical picture it would correspond to an expectation value of the Hamiltonian operator on the boundary spin-network state.

The drawback of the geometric approach is that it is difficult to find the right face and edge amplitudes, and ensure the gluing is consistent. 

\item A more well-founded complementary approach is based on the discretized path integral, viewed as a sum over (quantum) spacetime histories. The starting point is the BF path-integral measure:
\begin{equation}\label{eq:BF-partition-Z}
Z_{BF} \ = \ \int [dB] \ [d\om] \ e^{i \,\int\tr \, (B\wedge F[\om])} \ = \ \int [d \om]  \ \delta(F[\om]),
\end{equation}
which is well-defined in our discrete setting as
\begin{equation}\label{eq:BF-Z-discr}
Z_{BF} \ = \ \int\limits_{\mathfrak{h}^F \times H^E} \prod_f d B_f \ \prod_e dh_e \ \exp\bigg\{ i \,\sum_f \tr\,\bigg(B_f \overrightarrow{\prod_{e\subset f}}h_e\bigg)\bigg\}  \ \doteq \ \int\limits_{H^E} \prod_e dh_e \ \prod_f \delta\Bigg(\overrightarrow{\prod_{e\subset f}}h_e\Bigg).
\end{equation}
($F$ and $E$ denote the total number of faces and edges, respectively, of the 2-complex. The dot over equality sign forewarns that the second delta on the r.h.s. may appear, in general, depending on the actual group $H$ chosen.) Passing from the group elements $h_e$ to the representation category via the Plancherel theorem, one can recast \eqref{eq:BF-Z-discr} into the \eqref{eq:partition-Z} state-sum form. 
\end{itemize}

Considering the boundary and states on the induced graph $\Ga=\pa\Upsilon$, one immediately infers, quite generally, that both approaches lead to the kinematical Hilbert spaces spanned by $H$-spin networks for BF theory. The vertex amplitude is then obtained via evaluation of the boundary state on a flat connection. This picture is exact for gravity in 3d, where it is topological (and, thus, discretization independent). However, passing to 4d, the theory should be properly constrained, and this is where the various ambiguities arise. 

\paragraph{Quantizing the bivectors.} The graph $\Ga$ (cylindrical) state functional depends on connection by virtue of discrete holonomies of $H$. The bivectors act on $H$ as the right/left invariant vector fields, so that using Minkowski spacetime metric $\eta$ we can identify them with the elements of the the dual Lie algebra, carrying the natural Poisson structure: 
\begin{equation}
\begin{aligned}\label{eq:bivector-quant}
&\theta \ : & & \bigwedge^2\mathbb{R}^{3,1} && \ra && \mathfrak{so}(3,1)^\ast  && \\
&& & E_1\wedge E_2 && \mapsto && \theta(E_1\wedge E_2)(\mc B) \ := \ \eta(\mc B\rt E_1,E_2), \qquad \mc B\in \mathfrak{so}(3,1). &&
\end{aligned}
\end{equation}

For the following discussion, let us stick to convention, that the bivector $B$ refers, in general, to the symplectic structure -- namely, to the kinetic term, involving derivatives $d\om$ of the connection. It is the overall pre-factor in front of the curvature, and is promoted to tensor operator in certain representation (e.g. using techniques from geometric quantization). We also reserve the label $\Si=e\wedge e$ for the face bivector, for which the simplicity constraints are to be formulated, providing any $\Si$, satisfying them, with such an interpretation, conversely. This distinction is sensible in the light of the notorious fact that the correspondence between the two quantities is not unique, due to peculiar feature of the Hodge dual $\star: \bigwedge^2\mathbb{R}^4\ra\bigwedge^2\mathbb{R}^4$ in 4d, qualified as the Immirzi ambiguity in the quantization map:
\begin{equation}\label{eq:Immirzi-flipped-map}
B_f \ = \ \star\Si_f +\frac{1}{\ga} \Si_f \ \mapsto \ \hat{B}_f \qquad \stackrel{\ga^2\neq \pm 1}{\Longleftrightarrow} \qquad \Si_f \ \mapsto \ \left(\frac{1}{\ga}\mp\ga\right)^{-1}\left(\hat{B}_f-\ga \star \hat{B}_f\right).
\end{equation}
In other words, there are two independent invariant bilinear forms on $\mathfrak{so}(3,1)$, resulting in the Holst action, which is classically equivalent to the vacuum EC theory (on-shell), but may differ quantum mechanically.

The crucial step then is the implementation of a quantum version of the simplicity constraints at the level of state-sum for BF theory (with $\Si$ quantized as in~\eqref{eq:Immirzi-flipped-map})~\footnote{Since~\eqref{eq:Plebanski-cnstr} is invariant w.r.t. $\star$ but not $P_{\ga}=1+\frac{1}{\ga}\star$, it is somewhat perplexing that for finite $\ga$ both solution sectors (for $\Si$) lead to the Holst action for gravity with different effective parameters. This is not really the issue here, since the linear simplicity isolates the sectors in a more efficient way, irregardless of the value of $\ga$}:
\begin{equation}\label{eq:simplicity-quant}
\widehat{\mc C_\al[\Si]} \ \approx \ 0
\end{equation}
Depending on the first/second-class nature of the set \eqref{eq:simplicity-quant}, they should either annihilate the state functionals (\textit{\`{a} la} Dirac), or to be imposed weakly on matrix elements (\textit{\`{a} la} Gupta-Bleuler). This usually leads to restrictions on spin labels $j_f$ and/or intertwiners $\iota_e$ of the boundary Hilbert space states. 


\paragraph{Linear cross-simplicity.} As has been noted, the major ingredient in the new models is the \textit{linearization (partial) of the simplicity constraints}. It follows directly from the geometric meaning of conditions \eqref{diag-simplicity},\eqref{cross-simplicity}, which basically state that four triangles, belonging to the same tetrahedron $\tau_e$ and described by the area bivectors $\Sigma^{AB}_f$, lie in one hyperplane. 

In the original construction~\cite{EPRL-FK2008flipped2,EPRL-FK2008finiteImmirzi,EPRL-FK2008FK,EPRL-FK2008LS}, one associates a normal discrete 4d vector $\mc V^A_e,\ e=1,...,5$, to each of the five tetrahedra in the boundary of a 4-simplex (assume they are all timelike $\mc V_e\in \mathbb{H}^3_+\cong SL(2,\mathbb{C})/SU(2)$). The quadratic cross-simplicity~\eqref{cross-simplicity} is then replaced (rather \textit{ad hoc}) by the orthogonality requirement on bivectors: 
\begin{equation}\label{eq:linear-simplicity}
\text{(ii')}\qquad \forall f\supset e \ : \qquad \Si_f^{AB} \mc V_{Be}^{\phantom{B}} \ = \ 0 \qquad \Leftrightarrow  \qquad I_B(\mc V_e)\rt \Si_f \ = \ 0.
\end{equation}
The projector on the r.h.s. separates the ``boost'' components of $\Si$, which are co-aligned with $\mc V$:
\begin{equation}\label{eq:BR-projectors}
I_B(\mc V)^{AB,CD}\ :=\ \pm 2 \mc V^{[B}\eta^{A][C}\mc V^{D]}, \qquad I_R(\mc V)^{AB,CD}\ :=\ \eta^{A[C}\eta^{D]B} \mp 2 \mc V^{[B}\eta^{A][C}\mc V^{D]},
\end{equation}
from the ``rotational'' part, generating the conjugate $H_{\mc V} = h(\mc V)\rt SU(2)$ subgroup, which leaves $\mc V$ invariant. The gain in this new form of constraints is that it excludes the undesired $I^{\pm}$-solutions in~\eqref{eq:Plebanski-sectors}, leaving just the mix of the gravitational $II^{\pm}$-sectors, as well as degenerate one~\cite{Engle2011Plebanski-sectors} (which we do not consider here). It thus imposes stronger conditions than quadratic~\eqref{diag-simplicity} and~\eqref{cross-simplicity}, which then automatically follow.

We are making two remarks, calling the attention to normals $\mc V_e$, the prime interest of the present work. The first observation is that the symbolic notation of~\eqref{eq:simplicity-quant} is to be replaced with
\begin{equation}\label{eq:simplicity-quant+normal}
\widehat{\mc C_\al[\Si,\mc V]} \ \approx \ 0,
\end{equation}
in order to reflect the introduction of a new geometric objects, in addition to bivectors. The second comment concerns the relation between the two quantities on the constraint surface, namely: 
\begin{equation}\label{eq:linear-simplicity-solution}
\Si^{AB}_f \ = \ E^{[A}_1E^{B]}_2 \ = \ \frac{3!}{h_f^e} \star \left(\mc N^{\phantom{A}}_{f\supset e}\wedge\mc V^{\phantom{B}}_{e\phantom{f}}\right)^{AB} \qquad \text{or, conversely} \qquad  \mc V^A_e \ = \ \frac13 h_f^e \left(\star \Si^{AB}_f\right) \mc N^{\phantom{A}}_{B f\supset e} \, . 
\end{equation}
Here the edge vectors  $E_{1,2}\perp\mc N,\mc V$ of triangular face $S_f$ are orthogonal to the (spacelike, $\mc N^2=+1$) surface normal $\mc N$, lying within tetrahedron $\tau_e$: $\mc N\cdot\mc V=0$. In order for $|\mc V|^2\equiv\mc V^A\mc V_A$ to be the 3d volume (squared) of $\tau_e$, the proportionality coefficient is ought to be the height $h_f^e=(E_3\cdot\mc N)$ from the base $S_f$ to the apex. It appears from the above relations as if neither of $\Si,\mc V$ could be considered more ``fundamental'', since one can be expressed through the other and vice versa. The resolution of conundrum ultimately lies, not surprisingly, in the composite nature of quantities, both comprised of tetrad d.o.f. This standpoint will pave the way for our extension of field space in Sec.~\ref{sec:Poincare-BF}, and reformulation of~\eqref{eq:simplicity-quant+normal} in Sec.~\ref{sec:Poincare-Plebanski}. 


\paragraph{The EPRL map.} Up to this point, the Barbero-Immirzi parameter $\ga$ did not partake in the formulation of constraints and should be irrelevant for their geometric content. It, however, plays somewhat mysterious role in quantization and essential for comparison with canonical LQG theory. Let us briefly recap on the basic features of the quantum vertex amplitude which arise from the weak imposition of (the part of) constraints~\eqref{eq:linear-simplicity} on the group $H$ irreps that live on faces $f$ of the 2-complex $\Upsilon$, without delving too much into details though:

\begin{itemize}
\item The linear cross-simplicity~\eqref{eq:linear-simplicity} is imposed weakly in the gauge-fixed setting, i.e. for the standard normals -- either $\mc V^A_0=\de^A_0$ for spacelike (or $\mc V^A_3=\de^A_3$ for tetrahedra of mixed signature), characterizing the canonical embedding of $H_0=SU(2)$ (or $H_3=SU(1,1)$) into $H$. All the various techniques (such as vanishing matrix elements $\sim$ master constraint $\sim$ restriction of coherent state basis to those with the simple expectation values in the semi-classical limit) lead to the relation between 4d and 3d Casimirs:
\begin{equation}\label{eq:EPRL-cross-simpl}
\frac{C_H^{(2)}(\chi_f)}{2C_{\mc V}(j_{f\supset e})} \ \simeq \ \ga.
\end{equation}
This defines the embedding map for `spins' $j$ into decomposition of $SL(2,\mathbb{C})$ irreps $\chi_f$ w.r.t. little group. 

One nice feature of \eqref{eq:EPRL-cross-simpl} is that its exact implementation \cite{Alexandrov2010newSF-from-CovariantSU2} projects the spin-connection $\om$ in the holonomies~\eqref{eq:discrete-holonomies} to the (covariant lift of) Ashtekar-Barbero connection of LQG:
\begin{equation}\label{eq:projected-AB}
\pi^{(j)}\left(\om^{AB}_a \mc J^{(\chi)}_{AB} \right)\pi^{(j)} \ = \ {}^{(\ga)}A_a^I L_I^{(j)}, \qquad {}^{(\ga)}A^I \ = \ \frac{1}{2}\eps^{0I}_{\phantom{0I}JK}\om^{JK} - \ga \om^{0I},
\end{equation}
here $\pi^{(j)}$ projects on the $j$-irrep of the $SU(2)$ subgroup, and $L^I=\frac12\eps^{0IJK}\mc J_{JK}$ is the canonical generator of rotations in the corresponding representation.

\item The part of the linear simplicity~\eqref{eq:linear-simplicity} is first-class and imposed strongly. Taking into account~\eqref{eq:EPRL-cross-simpl}, it is equivalent to the (quadratic) diagonal simplicity~\eqref{diag-simplicity}, if the Barbero-Immirzi parameter is included. It relates the $SL(2,\mathbb{C})$ Casimirs:
\begin{equation}\label{eq:EPRL-diag-simpl}
\qquad \left(1 \pm\ga^2\right) C_H^{(2)}(\chi_f)-2\ga C_H^{(1)}(\chi_f) \ \simeq \ 0, 
\end{equation}
and puts restrictions on allowed `simple' irreps $\chi_f$.

\item The closure condition~\eqref{eq:closure-3d} translates into the requirement of the $H$-invariance of the amplitude and is ordinarily implemented through the group integration. It obviously encompasses the invariance w.r.t. the little group $H_0$ of the embedded $j$-states within the tensor product of simple representations, stacked at the tetrahedron $\tau_e$ bounded by the faces $S_f$. Thereby the EPRL embedding map is established:
\begin{equation}\label{eq:EPRL-map}
\Phi_\ga \ : \quad \mathrm{Inv}_{SU(2)}^{\phantom{1}} \,  \bigotimes_{f\supset e} \, j_f \ \longrightarrow \ \mathrm{Inv}_{SL(2,\mathbb{C})}^{\phantom{1}}  \, \bigotimes_{f\supset e}\, \chi_f,
\end{equation}
where we denoted the representation spaces with their corresponding labellings, for brevity. Hence the states in the boundary space are labelled by $SU(2)$ intertwiners glued into spin-networks. The last portion of the integration over the homogeneous space $H/H_0$ can be vied as summing over all possible gauge choices for the normals $\mc V \in H \rt \mc V_0$, and thus restoring the full Lorentz invariance at the vertex in the gauge-fixed model.
\end{itemize}

One clearly sees the subsidiary role of $\mc V$'s: in the construction of the model they are treated as ``unphysical'' gauge choice, which one can specify freely, and later ``erase'' this information. In effect, $\mc V$ allows one to reduce the problem of constraint imposition to the level of little group $H_0$, instead of operating directly on the covariant level of the full Lorentz group $H$. However, let us pinpoint some delicate issues, regarding these normals: 

\begin{itemize}

\item For instance, we know that the relative of the time-normal field explicitly appears as non-trivial lapse/shift components in the Lorentz-covariant canonical quantization of the 1st order action~\eqref{eq:Einstein-Cartan} with the Holst term. It is also an established fact that the boundary states are spanned by the projected spin networks~\cite{LivineDupuis2010lifting-SU2}:
\begin{equation}\label{eq:projected-spin-network}
\Psi_{\Ga=\pa\Upsilon}^{\phantom{\Ga}}\big(\big\{h_l\big\},\big\{\mc V_n\big\}\big)  \ = \ \left\langle \bigotimes_{l\in\pa f} \left( \pi^{(j_{t(l)})}_{\,\mc V} \ D^{\chi_l}_H (h_l) \ \pi^{(j_{s(l)})}_{\,\mc V}\right), \ \bigotimes_{n\in\pa e} \iota^{(n)}_{\mc V}\right\rangle,
\end{equation}
where these normals play a prominent role and are discretized naturally over the nodes. The state functionals are invariant w.r.t. the covariant Lorentz group action on both sets of variables:
\begin{equation}\label{eq:projected-transformation}
\Psi \big(\big\{h_l\big\},\big\{\mc V_n\big\}\big) \ = \ \Psi \big(\big\{U_{t(l)}^{-1}h_l^{\phantom{1}}U_{s(l)}^{\phantom{1}}\big\},\big\{U_n\rt\mc V_n\big\}\big), \qquad \forall \, U_n \in H.
\end{equation}

\item Historically, one of the incentives, which led to FK model~\cite{EPRL-FK2008FK}, was to solve the so-called ``ultra-locality'' problem with the BC amplitude -- namely, the apparent shortage in intertwiner d.o.f., which signified about the limited nature of correlations between neighbouring 4-simplices' geometries. 
On a more technical level, the resolution of identity, associated with the invariant vector space $X_e:=\mathrm{Inv}_H^{\phantom{1}}\big[\bigotimes_{f\supset e} \chi_f\big]$ at each edge of initial BF spin foam, rewritten in terms of coherent intertwiners (for a moment, $H=\mathrm{Spin}(4)\cong SU(2)\otimes SU(2)$):
\begin{equation}\label{eq:identity-resolution}
\mathbbm{1}_{X_e}^{\phantom{1}} \ = \ \bigotimes_\pm\int \prod_{f\supset e} d^2\mathbf{n}^\pm_{ef} \, d_{j^\pm_f} \int dh_{ve}^\pm\int dh_{v'e}^\pm \ h_{ve}^\pm\big|j^\pm_f,\mathbf{n}^\pm_{ef}\big\ket\big\bra j^\pm_f,\mathbf{n}^\pm_{ef}\big|\left(h_{v'e}^\pm\right)^\dagger ,
\end{equation}
is replaced by a projector, where summation is only over those states in the `simple' representations $j^+=j^-$, which solve the quantum cross-simplicity~\eqref{eq:linear-simplicity} (as expectation values). Specifically, the existence of a common $\check{u}_e\in SU(2)$ group element is inferred, representing 4d normal $\mc V_e$, which establishes the relation $\mathbf{n}^-=-\check{u}_e\rt\mathbf{n}^+$. 

The gluing of two 4-simplices -- via identifying first the geometries, corresponding to their common tetrahedron~$\tau_e$, and only then performing an integration -- takes into account the missing correlations between neighbouring vertices, sharing an edge. Whereas the unique Barrett-Crane intertwiner is obtained if one integrates separately at each vertex over (then decoupled) geometries. Arguably, the latter identification concerned only an internal 3d geometry of $\tau_e$, encoded in the spins and 3d normals $\{j,\mathbf{n}\}$, corresponding to the (canonically embedded) little group $H_0=SU(2)$.

The key point of the present work to treat normal $\mc{V}$ as truly independent geometric variable, characterizing the placement of 3d faces in 4d, creates some tension with the implementation of gauge invariance in the EPRL-FK vertex amplitude, if the above logic is extended by analogy to $\mc V$. Indeed, the dependence on the subsidiary variable $\check{u}_e$ is ``eaten'' by the follow-up $H$-group integration, performed independently at each vertex. The situation is quite similar to the BC intertwiner, so there still may be some d.o.f. left uncorrelated (even though if gauge).
 
\item A similar type of arguments have been put forward on the basis of the Lorentz-covariant canonical quantization endeavour~\cite{Alexandrov2008SF-from-CovariantLQG}. It has been argued that allowing an additional variable $\mc V$ remain unintegrated, the covariant transformation properties~\eqref{eq:projected-transformation} necessitate a relaxation of the closure condition~\cite{Alexandrov2008SF-cnstr-revisited}. The closure of the discrete bivectors is a too restrictive Gauss law, because the gauge transformations should act on the vector variables as well. In the preliminary Hamiltonian analysis of~\cite{GielenOriti2010Plebanski-linear}, the very same reason led authors to artificially enlarge the phase space by the fictitious momenta, corresponding to $\mc V$. In the Sec.~\ref{sec:Poincare-BF} we will see how our modification responds to both these objectives in quite a natural manner.


\end{itemize}

\subsection{The fate of the `volume' constraint}\label{ssec:volume-fate}
As discussed in Sec.~\ref{sec:classics}, the discretization of (quadratic) volume constraint employs several tetrahedra of the 4-simplex, hence it is usually thought of as consistency condition on time evolution (``secondary'' constraint). Indeed, \eqref{eq:volume-simplex} shows that if the cross-simplicity together with the 3d closure holds true for all tetrahedra, it does not matter which of the face bivectors are used to calculate the volume of the 4-simplex. Thus, it is not imposed explicitly in the quantum theory, once the former two are implemented. The same proof using the cable-wire diagrammatic representation of the 4-simplex amplitude shows that this indeed holds in  the quantum theory as well~\cite{Perez2013SF-review}, at least semi-classically.

We notice that the argument heavily relies on the combinatorics of the 4-simplex and does not necessarily extend to the generic case of arbitrary 2-complex. Explicitly, this appears already in the simple case of (hyper)cuboidal graph~\cite{BahrSteinhaus2016Cuboidal-EPRL}. There the plain ansatz, using the semi-classical substitute for the exact vertex amplitude, has been studied for the flat (no curvature) rectangular lattice. The expression for the amplitude is a straightforward KKL generalization of the Euclidean $\text{EPRL}_{\ga<1}$ model, in the FK representation using coherent states:
\begin{subequations}
\begin{gather}\label{eqs:cub-amplitude}
\mc A_v^\pm \ = \ \int \prod_e dh_{ve} \, e^{S^\pm[h_{ve}]} \ \sim \ \sqrt{\frac{(2\pi)^{21}}{\mathrm{det}\left(-\pa^2 S\big(\vec{h}_{\mathrm{c}}\big)\right)}} e^{S(\vec{h}_{\mathrm{c}})}, \quad j\ra \infty ,
\\
S^\pm[h_{ve}] \ = \ \frac{1\pm\ga}{2}\sum_{(ee')} 2j_{(ee')} \ln\bra-\mathbf{n}_{ee'}|h_{ve}^{-1}h_{ve'}^{\phantom1}|\mathbf{n}_{e'e}\ket.
\end{gather}
\end{subequations}
Here the summation goes over the (ordered) pairs $(ee')=(f\cap\pa\mc T_v)$ -- the (directed) links of a 6-valent combinatorial hypercuboidal boundary graph, and the data $\{j,\mathbf{n}\}$ in this symmetry reduced setting was chosen to represent (semi-classically) $\mathbb{R}^3$-cuboids: 
\begin{equation}\label{eq:cuboid-intertwiner}
|\iota\ket \ = \ \int d\check{u} \ \check{u}\rt\bigotimes_{i=1}^3|j_i,\mathbf{n}_i\ket|j_i',\mathbf{n}'_i\ket, \qquad j_i' \, \mathbf{n}_i' = - j_i \, \mathbf{n}_i,
\end{equation}
glued along their faces. $\pa^2 S$ denotes the Hessian matrix, evaluated at the critical point $\pa S\big(\vec{h}_{\mathrm{c}}\big)=\Re\, S\big(\vec{h}_{\mathrm{c}}\big)=0$.

\begin{figure}[!ht]
\center{\includegraphics[width=0.5\linewidth]{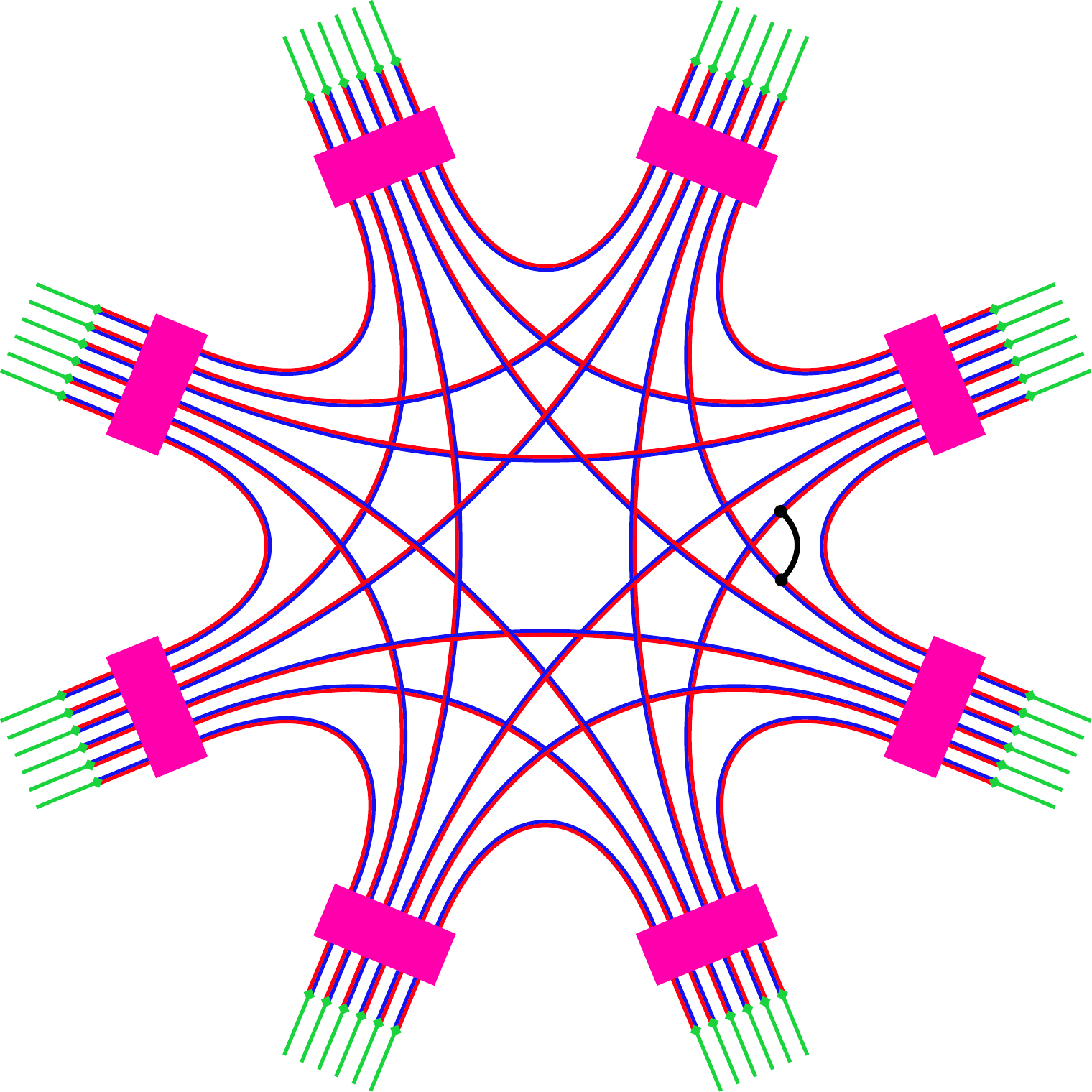}}
\caption{Diagrammatic representation of the hypercuboidal amplitude. The notation is as in~\cite{Perez2013SF-review}: lines (`wires') depict the $SL(2,\mathbb{C})$ representation matrices, and boxes (`cables') -- invariant projectors (group integrations). Note: apart from the graph's purely combinatorial properties, the relevance of the crossings for the 4-volume is highlighted in~\cite{BahrBelov2017VolumeSimplicity}.}
\label{fig:cable-wire-cub}
\end{figure}
It turns out that the 4-volume of a flat hypercuboid cannot be unambiguously ascribed to a vertex, using the prescription akin to~\eqref{volume-quadr} for 4-simplex, where its consistency is guaranteed by~\eqref{eq:volume-simplex}. If the rectangular lattice is \textit{geometric} (i.e. we are on the solution to simplicity constraints), it is characterized entirely in terms of its four edge lengths $E_i,\, i=t,x,y,z$, and the unique geometric 4-volume can be computed irregardless of the faces chosen $\tilde{V}_v:=E_tE_xE_yE_z=\Si_{tx}\Si_{yz}=\Si_{xy}\Si_{zt}=\Si_{xz}\Si_{yt}$, where each area is simply given by the product of the cooresponding edge lengths, e.g. $\Si_{xy}=E_xE_y$, etc. Instead, we get six arbitrary areas/spins~$j_{(ij)}$ which do not necessarily satisfy the above conditions. Indeed, if we try to proceed like in~\eqref{eq:volume-simplex}, starting with the expression $j_{xy}j_{zt}$ (depicted by a `grasping' on Fig.~\ref{fig:cable-wire-cub}) and applying the 3d closure for the spatial cuboid $\tau_t$, we end up with a tautological result: the contributions from parallel faces (bounding $\tau_t$ and the two adjacent anti-podal cuboids $\tau_i$, $\tau_{-i}$) enter with equal areas/spins but opposite signs $\eps\,\Si_{ij}\cdot\Si_{kt}=-\eps\,\Si_{ji}\cdot\Si_{kt}, \, i,j,k=x,y,z$, thus contracting each other in the sum~\footnote{Stronger, $\eps\,\Si_{xz}\cdot\Si_{zt}=\eps\,\Si_{yz}\cdot\Si_{zt}=0$ by the cross-simplicity.}, so we arrive at the dull equality $j_{xy}j_{zt}=j_{yx}j_{zt}$.

The essential ingredient of the EPRL construction, namely, that one could effectively replace the `volume' part of the simplicity by the 3d closure, is not valid for a higher valence. We encounter the problem that the model \emph{is not constrained enough} to complete the reduction from BF to gravitational theory. The measure of deviation is captured by the `non-geometricity' parameter, in this case:
\begin{equation}\label{eq:non-geometricity}
\varsigma \ = \ \begin{pmatrix}
           j_{xy}j_{zt}-j_{xz}j_{yt} \\
           j_{xz}j_{yt}-j_{xt}j_{yz} \\
           j_{xy}j_{zt}-j_{xt}j_{yz}
         \end{pmatrix}.
\end{equation}
The numerical studies of~\cite{BahrSteinhaus2016Cuboidal-EPRL} show that the non-geometric configurations with $\varsigma\neq0$ \textit{do} generically contribute to the path-integral, although their impact might be exponentially suppressed. The dumping is controlled by the width of the Gaussian -- the effective ``mass'' term $m^2_\varsigma(\al) \approx 2\al -1>0$ for $\al\gtrsim 0.5$, which depends crucially on the parameter $\al$ in the choice of the face amplitude $\mc A_f^{(\al)} = \big((2j_f^++1)(2j_f^-+1)\big)^\al$. Reassuringly, in the same range of~$\al$ indications were given for the tentative continuum limit in the form of a phase transition, with the restoration of the (remnant) diff-invariance. This led authors to suggest that the allowed freedom in the face amplitude might be restricted on physical grounds, for one should definitely obtain geometric states in the classical limit.

\paragraph{Fully linear treatment.} Naturally, the 2 missing constraints to impose in this elucidating example are $\varsigma=0$, however, it is unevident how to proceed in the most general case. The simplicial `4-volume constraint' makes little sense here, unless appended with some additional requirements (as we tentatively propose in~\cite{BahrBelov2017VolumeSimplicity}, based on the certain type of graph invariants). Being the part of Pleba\'{n}ski's \emph{quadratic} formulation, it is also inorganic to the model built on linear constraints. An alternative fully \textit{linear} formulation was put forward in~\cite{GielenOriti2010Plebanski-linear}, providing both the continuum version of the cross-simplicity~\eqref{eq:linear-simplicity}, as well as the linearized counterpart for the `volume' constraints of the form~\eqref{eq:simplicity-quant+normal}. It introduces the basis of 3-forms $\vartheta^A$, whose discretization naturally associates 4d normal vectors
\begin{equation}\label{eq:discrete-normals}
\mc V^A_e \ = \ \int_{\tau_e} \vartheta^A
\end{equation}
to tetrahedra~$\tau_e$. Although, strictly speaking, neither Pleba\'{n}ski, nor Gielen-Oriti's linear version were ever formulated beyond triangulations, let us extrapolate the latter to our cuboidal setting, like we did with 4-volume. 

\begin{figure}[!t]
\center{\includegraphics[width=0.7\linewidth]{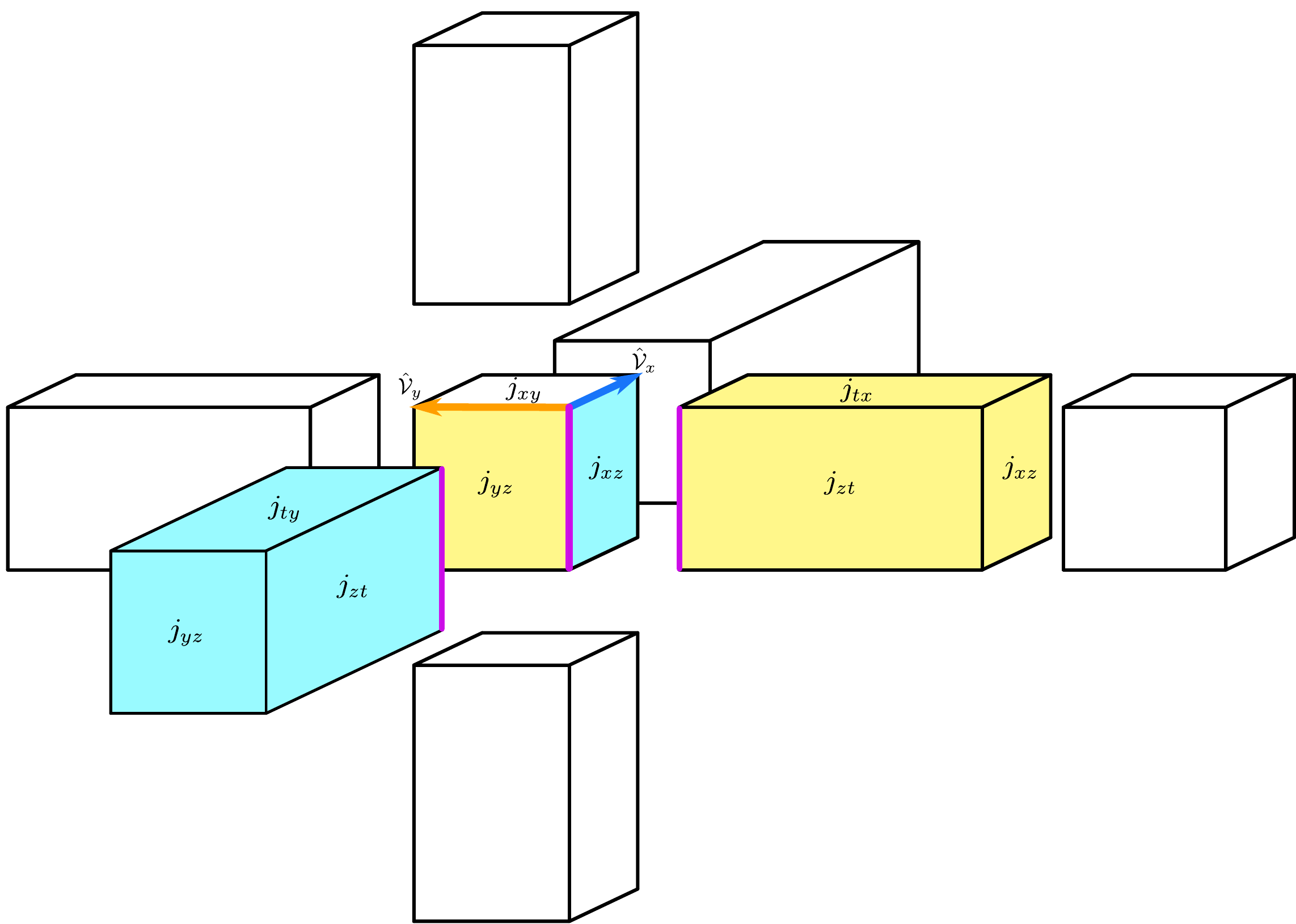}}
\caption{Elements of linear volume constraint: edge $E_z$ (purple) is shared by the two faces of the spatial cuboid in the middle, bounding it from the cuboids, orthogonal to $x$ (blue) and $y$ (yellow) directions, respectively. Colour scheme corresponds to $\Si\cdot\mc V$ pairing.}
\label{fig:hypercube}
\end{figure}

Analogously to 4-simplex, every vertex is the source of four edge vectors, each of which is shared by exactly three cuboids, intersecting along three faces, respectively. Their directions can be identified with that of the eight cuboidal normals $\hat{\mc V}_i$ (e.g. aligned with $\hat{x}_i$ and $\hat{x}_i' = -\hat{x}_i$ of standard cartesian grid). However, their norms are considered as free parameters (like surface areas $j$). The bivector data~\footnote{We exclude the Immirzi parameter from consideration and put $B=\star\Si$. From here on all the following discussion is purely classical.} is restricted to satisfy
\begin{equation}\label{eq:hypercube-bivectors}
(B_{ij}+B_{i'j})\cdot\hat{\mc V}_j \ = \ 0,
\end{equation}
that is opposite faces of 3d cuboid are required to be parallel and of equal area. Indeed, this is precisely the boundary data for coherent intertwiners~\eqref{eq:cuboid-intertwiner}, if 3d normals are defined as $(j\hat{\mc N})_i= B_{ij}\cdot\hat{\mc V}_j$. Pick up some vertex, and fix the edge $\hat{\mc V}_l$, ejecting from it. For every such pair, via straightforward transfer from~\cite{GielenOriti2010Plebanski-linear}, one writes down the two linearised volume (vector) constraints:
\begin{equation}\label{eq:linear-volume-simplicity}
\mc V_i\cdot \Si_{il} \ = \ \mc V_j\cdot \Si_{jl} \ = \ \mc V_k\cdot \Si_{kl} \qquad \forall \ i,j,k\neq l.
\end{equation}
We have one equality per each 3d cuboid in the cycle, sharing the given edge as a hinge, at which the two consecutive 2d faces intersect. In other words, the volume constraint is discretized at the edge-cuboid pairing (see figure~\ref{fig:hypercube}). Because of orthogonality, all the terms are proportional to $\hat{\mc V}_l$. Let us regroup it at each vertex via formation of sums over edges, orthogonal to one particular direction $\hat{\mc V}_k$:
\begin{equation}\label{eq:reduced-simplicity}
\sum_{\{i,j\}\neq k} \mc V_{i}\cdot (\star B_{jk}) \ = \ 0 \qquad \forall \, k.
\end{equation}
(Writing $\star B_{zt}=\Si_{xy}$, etc., is just the convenient relabelling of face by its two orthogonal directions, corresponding to cuboids adjacent along $S_f$.) Each individual term in the sum is proportional to $\hat{\mc V}_{l\neq k}$, and by linear independence in 3d subspace $\perp\hat{\mc V}_k$ we can rewrite pre-factors in the form of proportion:
\begin{equation}\label{eq:proportion}
\frac{|\mc V_i|}{|\mc V_j|} \ = \ \frac{|B_{ik}|}{|B_{jk}|} \qquad \forall \, i,j\neq k,
\end{equation}
the third equality being the consequence of the other two. Taking another such identity for $l\neq k$ and the same $i,j$ we arrive at $|B_{ik}||B_{jl}|=|B_{il}||B_{jk}|$ -- the familiar geometricity conditions. 

Our proof, however, showed something more than that. The volumes $|\mc V_i|$ of three cuboids in the cycle around selected edge $\hat{\mc V}_k$ relate as areas of their bases $|B_{ik}|$. The proportionality coefficient $|E_k| = |\mc V_i|/|B_{ik}|$ is independent of $i$ and we infer the existence of heights/edge lengths, invariantly defined in terms of (independent) $|\mc V_i|, |B_{ik}|$ from any of three cuboids, sharing this edge. The shapes of faces are obviously matching. We essentially exploited self-reciprocal nature of orthogonal lattice, allowing an identification of edges with $\mc V_i$ up to constants. Had we allowed fluctuations in directions $\hat{\mc V}_i$, the two notions would separate, and multiple angles between the lattice and its dual would intervene the formulas.

\paragraph{Closure of 4d normals.} If we agree to call ``non-local'' or ``dynamical'' either the quantities or relations, where objects from several elements of 3d boundary are involved, then the constraint~\eqref{eq:reduced-simplicity} is necessarily non-local. Indeed, it collects volumes of three cuboids at their intersection edge, multiplied by the areas that they cut at the beginning of the edge. Gielen and Oriti demonstrated in the case of the 4-simplex that these non-trivial equations may be replaced by another non-local condition of \emph{4d closure} of normals: 
\begin{equation}\label{eq:closure-4d}
\sum_{e\supset v} \mc V^A_e \ = \ 0 \qquad \forall \, v,
\end{equation}
together with the usual cross-simplicity~\eqref{eq:linear-simplicity} and 3d closure of bivectors~\eqref{eq:closure-3d} at each tetrahedron $\tau_e$ (i.e. ``local'', or ``kinematical'' constraints). 

This result reminds of the \textit{Minkowski's theorem}, extended to 4d. Recall that the latter asserts the existence of the unique (up to congruence) flat convex polytope, if the set of its face normals $\{\mc V_e\}$ is given, satisfying~\eqref{eq:closure-4d} (cf.~\cite{BianchiDonaSpeziale2011Polyhedra} for the 3d case). The non-trivial part is to demonstrate that it is compatible with characterization via simplicity constraints. We did not find the way to prove the analogous statement for hypercuboid. The complication, arising in this case, we relate to the presence of well-separated boundary elements (e.g. ``past'' and ``future'' cuboids in the foliation picture). Whereas in 4-simplex, every tetrahedron shares at least one common point with any other.


Returning to constrained-BF framework of Spin Foam models, our inspection highlights the following fact. In order to complete the reduction from the topological theory to gravity, using the linear formulation~\eqref{eq:simplicity-quant+normal} with independent normals, some additional requirements have to be met. This may be either `linear volume' constraint, or its equivalent. The obvious rewriting of the norms $|\mc V_l|$, using~\eqref{eq:proportion}, as
\begin{equation}\label{eq:volume-unique}
|E_i| |B_{il}| \ = \ |E_j| |B_{jl}| \ = \ |E_k| |B_{kl}| \qquad \forall \, i,j,k\neq l,
\end{equation}
leads to useful (re)interpretation of constraints~\eqref{eq:linear-volume-simplicity}. Namely, they indirectly characterize edge lengths through compatibility both with $B$'s \emph{and} $\mc V$'s, such that the unique (metric) volume, equal to $|\mc V_l|$, can be ascribed to each cuboid. After this identification the normals have to satisfy closure~\eqref{eq:closure-4d}.

We find it convenient to switch to a simpler variables, leading to a more manageable set of relations. Instead of volumes $\mc V$, we therefore propose to use directly the edge vectors $E$, as suggested by~\eqref{eq:volume-unique}. Unlike the original Gielen-Oriti's~\eqref{eq:linear-volume-simplicity}, the new version can be stated locally at the level of 3d polyhedra. We study the continuum formulation behind it in Sec.~\ref{sec:Poincare-Plebanski}, proving the equivalence with the usual simplicity constraints $B=\star e\wedge e$.

As another option, one could prefer to impose the closure on 4d normals $\mc V$, at least in the case of 4-simplex. Ordinarily, the 3d closure~\eqref{eq:closure-3d} of bivectors can be thought of as the result of integration of the continuous Gauss law $\nabla^{(\om)}_{[c}\Si^{AB}_{ab]}=0$ over a 3-ball triangulation (using the Stokes' theorem and the flatness of connection).
There is, however, the problem with such field-theoretic interpretation of~\eqref{eq:closure-4d}, because the corresponding dynamical law $\nabla^{(\om)}_{[a}\vartheta^A_{bcd]}=0$ for independent 3-forms is absent in the continuum theory. This flaw is fixed in the Sec.~\ref{sec:Poincare-BF}. In accord with our choice of variables, we first establish the link with the condition of vanishing torsion for tetrads $e$, and then implement it via Lagrange multipliers, prior to any imposition of simplicity constraints.


\section{The extended BF action with the frame field}\label{sec:Poincare-BF}

The 3d closure condition~\eqref{eq:closure-3d} represents the discrete Gauss law of the BF theory. We are looking for a way to accommodate the discrete 4d closure~\eqref{eq:closure-4d} in a similar fashion within a continuum theory through the equations of motion. We restrict ourselves here with the modification of the unconstrained topological theory, postponing the discussion of the simplicity constraints to the next section.

Recall that the 3-forms appearing in~\eqref{eq:discrete-normals} are eventually related to tetrads:
\begin{align}\label{eq:3-forms}
\vartheta^A \ & = \ \frac{1}{3!}\vartheta^A_{bcd}\, dx^b\wedge dx^c\wedge dx^d \nonumber \\ 
& = \ \frac{1}{3!}\eps^A_{\phantom{A}BCD}\, e^B_b e^C_c e^D_d \, dx^b\wedge dx^c\wedge dx^d \ = \ \frac{1}{3!}\eps^A_{\phantom{A}BCD}\, e^B\wedge e^C\wedge e^D.
\end{align}

By the Hodge duality (see~\ref{app1}), we may prefer an alternative parametrization and use directly the tetrad field  $e^A_a$, without any loss of generality -- the number of components is the same as $\vartheta^A_{bcd}$. In order to identify the appropriate kinetic term, suppose that a piecewise-flat cell complex $\bigcup_v \mc T_v\simeq \mc M$, approximating the target spacetime manifold, is given. Express the l.h.s. of~\eqref{eq:closure-4d} through Stokes' theorem:
\begin{equation}\label{eq:Stokes}
\sum_{e\supset v} \mc V^A_e \ = \ \int_{\pa\mc T_v} \vartheta^A \ = \ \int_{\mc T_v} d\vartheta^A, \qquad \text{where} \quad \pa\mc T_v = \bigcup_{e\supset v}\tau_e
\end{equation}
-- the boundary of the 4-simplex $\mc T_v$ dual to the vertex $v$ (more generally, any 4-polyhedron). One possibility is to introduce it with the Lagrange multipliers (and covariant derivative for the arbitrary smooth manifold):
\begin{gather}
\int \mathpzc{b}_A^{\phantom{A}} \, D \vartheta^A \ = \ \int \frac12 \mathpzc{b}_A^{\phantom{A}}\eps^A_{\phantom{A}BCD}\, e^B\wedge e^C\wedge D e^D \ \equiv \ -\int \be_A^{\phantom{A}}\wedge T^A, \nonumber\\
\text{where} \quad \be_A^{\phantom{A}} \ := \ \frac12 \eps_{ABCD}^{\phantom{ABCD}} \, \mathpzc{b}^B e^C\wedge e^D \ \equiv \ \frac12 \be_{Abc}^{\phantom{A}}\, dx^b\wedge dx^c, \quad \be^A_{cd}\ =\ \eps^{AB}_{\phantom{AB}CD}\, \mathpzc{b}_B^{\phantom{B}}e^C_c e^D_d \label{eq:T-multiplier}
\end{gather}
-- the 2-form taking values in $\mathbb{R}^{3,1}$, whilst $\mathpzc{b}_A$ are scalar functions (0-forms). The vanishing of torsion $T^A:=De^A$ in the first line is sufficient for $D\vartheta^A=0$ (4d closure), and we thus further choose to take generic~$\be_A$ as an independent Lagrange multipliers (conjugate to $e^A$), imposing the stronger condition $De^A=0$~\footnote{It can be noted that an $n$-form with vanishing exterior derivative and appropriate internal indices can be given a geometric interpretation as describing $n$-simplices closing up to form an $(n+1)$-simplex. $De^A=0$ may be seen as another instance where the edges of the 2d triangle add up to zero.}. There is still a possibility left for exploration if we had chosen independent $\mathpzc{b}_A$ and weaker equation $D\vartheta^A=0$~\footnote{In this regard, it would be interesting to compare the two constraints and to identify the sets of compatible geometric configurations in each case (especially in the light of allowed `conformal shape-mismatch' in the recent work~\cite{Dona-etal2017KKL-asumpt}).}.

Before studying gravity per se, let us examine first the simplifying theory, which one gets by modifying correspondingly the pure BF action. Thereby, we suggest to consider the following unconstrained theory:
\begin{equation}\label{eq:BF-Poincare}
S_{0}[B,\be,\om,e] \ = \ \int B_{AB}^{\phantom{AB}}\wedge F^{AB} + \be_A^{\phantom{A}}\wedge T^A.
\end{equation}
The main assumption made about the frame field is that it is non-degenerate, and the matrix $e^A_a$ is invertible. In the form of the action we may recognize the special case of the theory~\cite{MontesinosCuesta2008BF-Cartan-Ndim,MontesinosCuesta2007BF-Cartan}, designed to obey the Cartan's (also Bianchi's) structure equations (which are the kinematical basis for Riemannian geometry). One could anticipate and announce~\eqref{eq:BF-Poincare} as the Poincar\'{e} BF theory -- this assertion is made precise, as we develop in this section the structure of the theory in full detail. 

\subsection{Gauge symmetries: the Lagrangian approach}

The theories in physics are defined by their symmetries, which are usually presupposed. To find them out from the given action principle we follow the route, converse to the renowned Noether's 2nd theorem. Namely, the gauge symmetries (by which we understand the dependence of the dynamics on an arbitrary functions) manifest themselves through the differential identities among equations of motion. Calculate the Euler-Lagrange derivatives (putting the results into convenient language of forms):

\begin{subequations}\label{eq:BF-Poincare-variations}
\begin{align}
&& \frac{\de S_0}{\de B_{AB}}&&& := && \frac{1}{2!}\eps_{abcd} \frac{\de S_0}{\de B_{ABab}} dx^c \wedge dx^d &&= && F^{AB} , && \\
&& \frac{\de S_0}{\de \be_A} &&& := && \frac{1}{2!}\eps_{abcd} \frac{\de S_0}{\de \be_{Aab}} dx^c \wedge dx^d &&= && T^A , && \\
&& \frac{\de S_0}{\de \om_{AB}} &&& := && \frac{1}{3!}\eps_{abcd} \frac{\de S_0}{\de \om_{aAB}} dx^b \wedge dx^c \wedge dx^d &&= && D B^{AB} - e^{[A}\wedge \be^{B]} , && \label{eq:variation-omega}\\
&& \frac{\de S_0}{\de e_A} &&& := && \frac{1}{3!}\eps_{abcd} \frac{\de S_0}{\de e_{aA}} dx^b \wedge dx^c \wedge dx^d &&= && D \be^A. &&
\end{align}
\end{subequations}

Setting variations to zero (with some fixed boundary conditions), we obtain the field equations, to which every physical motion must satisfy:

\begin{equation}\label{eq:BF-Poincare-e.o.m.}
\frac{\de S_0}{\de B}, \frac{\de S_0}{\de \be}, \frac{\de S_0}{\de \om}, \frac{\de S_0}{\de e} \ = \ 0.
\end{equation}
The torsion now manifestly vanishes (as well as the curvature), and one recognizes in the third line~\eqref{eq:variation-omega} the generalized covariant conservation of $B$, i.e. the lifted 3d closure of bivectors in the discrete (as promised earlier).

It is now straightforward to derive the corresponding differential relations between functional derivatives~\eqref{eq:BF-Poincare-variations}: 
\begin{subequations}\label{eq:differential-identities}
\begin{align}
&&& \mc I_{AB}^{(1)}  && := &&  D \frac{\de S_0}{\de B^{AB}} \ = \ 0 , &&& \\
&&& \mc I_A^{(2)}  && := &&  D \frac{\de S_0}{\de \be^A} - e^B \wedge \frac{\de S_0}{\de B^{AB}} \ = \ 0 , &&& \\
&&& \mc I_{AB}^{(3)}  && :=  && D \frac{\de S_0}{\de \om^{AB}} - 2 B_{[A}^{\phantom{[B}C}\wedge\frac{\de S_0}{\de B^{B]C}}  - \be_{[A}^{\phantom{[}}\wedge \frac{\de S_0}{\de \be^{B]}} - e_{[A}^{\phantom{[}}\wedge \frac{\de S_0}{\de e^{B]}} \ = \ 0 , &&& \\
&&& \mc I_A^{(4)}  && := &&  D \frac{\de S_0}{\de e^A} - \be^B\wedge\frac{\de S_0}{\de B^{AB}} \ = \ 0 , &&&
\end{align}
\end{subequations}
which are vanishing identically (off-shell), without using the equations of motion (only their form). To obtain the above relations we used the 2nd and 1st Bianchi's identities, i.e. the commutation of covariant derivatives:
\begin{equation*}
DF^{AB}  =  0, \qquad DDe^A  =  F^A_{\phantom{A}B}\wedge e^B, \qquad DD\be^A  =  F^A_{\phantom{A}B}\wedge \be^B, \qquad DDB^{AB}  =  F^A_{\phantom{A}C}\wedge B^{CB} + F^B_{\phantom{B}C}\wedge B^{AC}.
\end{equation*} 
From the complete set of independent differential identities~\eqref{eq:differential-identities} we now form a generic linear combination, which is identically equal to zero, and integrate by parts:
\begin{gather*}
\int \Xi^{AB}\wedge \mc I_{AB}^{(1)} + \xi^A\wedge \mc I_A^{(2)} + \mathpzc{U}^{AB}\mc I_{AB}^{(3)} + \mathpzc{u}^A\mc I_A^{(4)} \\
= \int \left(D\Xi^{AB} - \xi^A\wedge e^B + \mathpzc{U}^A_{\phantom{A}C}B^{CB}+\mathpzc{U}^B_{\phantom{B}C}B^{AC} - \mathpzc{u}^A\be^B\right)\wedge \frac{\de S_0}{\de B^{AB}} \\
+\left(D\xi^A+\mathpzc{U}^A_{\phantom{A}B}\be^B\right)\wedge \frac{\de S_0}{\de \be^A} - D\mathpzc{U}^{AB}\wedge \frac{\de S_0}{\de \om^{AB}} - \left(D\mathpzc{u}^A -\mathpzc{U}^A_{\phantom{A}B}e^B\right)\wedge \frac{\de S_0}{\de e^A},
\end{gather*}
for some arbitrary coefficient functions $(\mathpzc{U}^{AB},\mathpzc{u}^C)$ and 1-forms $(\Xi^{AB},\xi^C)$. The transformations that leave the action invariant, up to divergence, are readily seen:
\begin{subequations}\label{eq:BF-Poincare-gauge-transform}
\begin{align}
\de \om^{AB} \ & = \ -  d\mathpzc{U}^{AB} + f^{AB}_{CD,EF}\, \om^{CD}\mathpzc{U}^{EF}, \label{eq:Poincare-connection-transform-om} \\ 
\de e^A \ & = \ - d\mathpzc{u}^A + f^A_{CD,B}\left(\om^{CD}\mathpzc{u}^B-\mathpzc{U}^{CD}e^B\right), \label{eq:Poincare-connection-transform-e}\\
\de B^{AB} \ & = \ -f^{AB}_{CD,EF}\mathpzc{U}^{CD}B^{EF}-\mathpzc{u}^{[A}\be^{B]} +D\Xi^{AB} - \xi^{[A}\wedge e^{B]}, \label{eq:null-B-transform} \\ 
\de \be^A \ & = \ -f^A_{CD,B}\mathpzc{U}^{CD}\be^B + D\xi^A, \label{eq:null-T-transform} \\
\text{where} \quad f^{AB}_{CD,EF} \ : & = \ \eta^{\phantom{[}}_{E[C}\de^{[A}_{D]}\de^{B]}_{F\phantom{]}} - \eta^{\phantom{[}}_{F[C}\de^{[A}_{D]}\de^{B]}_{E\phantom{]}}, \quad  f^A_{CD,B} \ := \ \eta_{B[C}^{\phantom{AB}}\de^A_{D]}. \label{eq:Poincare-structure-constants}
\end{align}
\end{subequations}

We may combine the connection $\om\in\mathfrak{so}(3,1)$ and the gauge potential of translations $e$ into a single (Cartan) connection~\cite{Wise2010Cartan-geometry} of the Poincar\'{e} gauge group~\footnote{This is quite appealing from another point of view that the usual non-degeneracy condition on $\mr{det}\,e^A_a\neq0$, which is unclear how to achieve in practise, is also the requirement for the $\mathfrak{g}$-connection to be a Cartan connection.}:
\begin{equation}\label{eq:Poincare-connection}
\mathfrak{so}(3,1)\ltimes\mathfrak{p}^{3,1}\ \ni\ \varpi \ := \ \om + e \ = \ \om^{AB}\mc J_{AB} + e^C\mc P_C,
\end{equation}
whose generators $\mc J_{AB},\mc P_C$ satisfy the algebra (of which~\eqref{eq:Poincare-structure-constants} are structure constants):
\begin{align}
\begin{split}\label{eq:Poincare-algebra}
[\mc J_{AB},\mc J_{CD}] \ & = \ i  \left(\eta_{C[A}\mc J_{B]D}-\eta_{D[A}\mc J_{B]C}\right), \\ 
[\mc J_{AB},\mc P_C] \ & = \ i \, \eta_{C[A}\mc P_{B]}, \\ 
[\mc P_A,\mc P_B] \ & = \ 0.
\end{split}
\end{align}
A local gauge transformation, taking values in $G=SO(3,1)\ltimes\mathbb{P}^{3,1}$, can be split into
\begin{equation}\label{eq:Poincare-split}
g(x) \ = \ u(x)U(x), \qquad u \ = \ e^{-i\mathpzc{u}\mc P}, \qquad U \ = \ e^{-i\mathpzc{U}\mc J},
\end{equation}
s.t. $u(x)$ changes the zero section, i.e. changes the local identification of points of tangency at each spacetime event. The transformation law for the connection is then
\begin{equation}\label{eq:Poincare-gauge-transform}
\varpi \ \rightarrow \ \varpi' \ = \ g^{-1}(\varpi+d)g \ = \ U^{-1}u^{-1}(\om + e) u U + U^{-1}u^{-1}(du)U + U^{-1}dU,
\end{equation}
whose infinitesimal form is~\eqref{eq:Poincare-connection-transform-om}, \eqref{eq:Poincare-connection-transform-e}. The combined curvature transforms in the adjoint representation:
\begin{subequations}\label{eq:Poincare-adjoint-transform}
\begin{gather}
\mathpzc{F}[\varpi] \ \rightarrow \ \mathpzc{F}' \ = \ g^{-1}\mathpzc{F}g \ = \ U^{-1}u^{-1}(F+T)u U, \\
\intertext{or, infinitesimally:}
\de F^{AB} \ = \ f^{AB}_{CD,EF}F^{CD}\mathpzc{U}^{EF}, \qquad \de T^A \ = \ f^A_{CD,B} \left(F^{CD}\mathpzc{u}^B - \mathpzc{U}^{CD} T^B\right).
\end{gather}
\end{subequations} 

However, the analogous quantity, comprised of the conjugate variables  
\begin{equation}\label{eq:Poincare-B-field}
\mathfrak{so}(3,1)\ltimes\mathfrak{p}^{3,1}\ \ni\ \mathpzc{B} \ := \ B + \be \ = \ B^{AB}\mc J_{AB} + \be^C\mc P_C, 
\end{equation}
demonstrates slightly different behaviour~\eqref{eq:null-B-transform},\eqref{eq:null-T-transform}, compensating for~\eqref{eq:Poincare-adjoint-transform} in order to make the action invariant. The inhomogeneity is transferred from $\be$ to $B$ part. For instance, if $\be$ is of the form~\eqref{eq:T-multiplier}, the corresponding addition would be $-\de_\mathpzc{u}\star B^{AB}=\frac12 (\mathpzc{u}\cdot \mathpzc{b}) e^{[A}\wedge e^{B]} + \mathpzc{b}^{[A}e^{B]}\wedge(\mathpzc{u}\cdot e)$. We are tempted to interpret the latter as being responsible for mixing up the simple bivectors -- the symmetry which is usually broken in order to obtain General Relativity. In addition to the usual (internal) gauge transformations, the action~\eqref{eq:BF-Poincare} is also invariant w.r.t. the shifts
\begin{equation}\label{eq:null-topological-symmetry}
\de\varpi \ = \ 0, \qquad \de \mathpzc{B} \ = \ \left(D\Xi^{AB}-e^A\wedge\xi^B\right)\mc J_{AB} + D \xi^C\mc P_C \ \equiv \ D_\varpi(\Xi+\xi) ,
\end{equation}
which extends the usual `topological' BF symmetry. In fact, it is always possible to gauge away any local d.o.f.; however, if $\mc M$ is topologically non-trivial, $\varpi$ and $\mathpzc{B}$ can have non-trivial solutions globally (hence the name).

Despite the non-conventional form of~\eqref{eq:null-B-transform}, we still have obtained the right connection~\eqref{eq:Poincare-connection} and the algebra~\eqref{eq:Poincare-algebra}. So that allows us to conclude that we have constructed a topological theory of the BF type for the Poincar\'{e} gauge group. We may guess that the departure of $\mathpzc{B}$ transformation properties from the adjoint~\eqref{eq:Poincare-adjoint-transform} is a forced decision, due to the degenerate nature of the Killing form: for the semidirect product algebra with the abelean ideal of translations it reduces to $\Tr \big[\mathrm{ad}_{(U,u)}\circ\mathrm{ad}_{(V,v)}\big]=2\,\Tr \big[\mathrm{ad}_U\circ \mathrm{ad}_V\big]\propto\eta_{A[C}^{\phantom{b}}\eta_{D]B}^{\phantom{b}}=-\frac14 f^{GH\phantom{]}}_{AB,EF}f^{EF\phantom{]}}_{CD,GH}$. Thus, the temptation to read the expression in~\eqref{eq:BF-Poincare} as the Cartan-Killing pairing, similar to the $\int\bra \mathpzc{B}\wedge \mathpzc{F}\ket$ with semisimple group, faces obstacles. The contraction for $\be - T$ is performed with the $\mc E$-bundle metric, e.g. obtained from the $\mc J-\mc P$ vector couplings. We find this to be an interesting arena for the imposition of simplicity constraints.

\paragraph{On the diffeomorphisms.} Note that the topological shifts~\eqref{eq:null-topological-symmetry}, parametrized by Lie-algebra valued 1-forms, leave the spacetime indices intact, as another manifestation of the topological character of the theory. On the contrary, the diffeomorphisms of $\mc M$ are generated by the (horizontal) vector fields $\zeta \in \Ga(T\mc M)$, its infinitesimal action on tensors being given by the Lie derivative. Hence, the question of relation between the two symmetries is a subtle one. To clarify the issue one would have to give a definite answer to the following questions: Whether the diffeomorphism transformations constitute an \emph{independent} symmetry? Since the action functional is constructed in a coordinate-free manner, one expects it to be invariant under diffeomorphisms. However, among the derived Noether identities we do not immediately find the respective one with the right tensorial structure, which one could contract with the parameter vector field. 

One typically finds in the literature~\cite{Buffenoir-etal2004Hamiltonian-Plebanski}, the following reply to our query. Of course, we are free to form various linear combinations of~\eqref{eq:differential-identities}, allowing for the field-dependent coefficients. In such a manner, the following identity can be constructed from the basic ones: 
\begin{equation}\label{eq:null-diffeo-identity}
\int \left(B^{AB} \lrcorner \, \zeta\right) \wedge \mc I_{AB}^{(1)} + \left(\be^A \lrcorner \, \zeta \right) \wedge \mc I_A^{(2)} +\left(\om^{AB}\lrcorner \, \zeta \right) \mc I_{AB}^{(3)} + \left( e^A \lrcorner \, \zeta\right) \mc I_A^{(4)} \ = \ 0,
\end{equation}
where we contract with the vector field $\zeta = \zeta^a\pa_a$ using the interior product~\eqref{app:interior-product}. This leads to the transformations which are seemingly related to diffeomorphisms: 
\begin{equation}\label{eq:gauge-diffeos-relation}
\begin{aligned}
\mc L_\zeta^{\phantom{0}} \om_{AB}^{\phantom{0}} \ & = \ - \de \om_{AB}^{\phantom{0}} + \left(\zeta \lrcorner \, \frac{\de S_0}{\de B^{AB}}\right), \\
\mc L_\zeta^{\phantom{0}} e_A^{\phantom{0}} \ & = \ - \de e_A^{\phantom{0}} + \left(\zeta \lrcorner \, \frac{\de S_0}{\de \be^A}\right), \\
\mc L_\zeta^{\phantom{0}} B_{AB}^{\phantom{0}} \ & = \ - \de B_{AB}^{\phantom{0}} + \left(\zeta \lrcorner \, \frac{\de S_0}{\de \om^{AB}}\right), \\
\mc L_\zeta^{\phantom{0}} \be_A^{\phantom{0}} \ & = \ - \de \be_A^{\phantom{0}} + \left(\zeta \lrcorner \, \frac{\de S_0}{\de e^A}\right),
\end{aligned}
\end{equation}
with gauge parameters as in~\eqref{eq:null-diffeo-identity}; however, the equivalence between the two being valid only on-shell. 

This is peculiar, since it tights the transformation properties to the fact, whether the equations of motion are satisfied. We stress, in this regard, that the diffeos of the conventional 2nd order metric GR are the elementary gauge symmetries in the above sense, both at the Lagrangian and Hamiltonian levels~\cite{KiriushchevaKuzmin2011canonical-GR-myths}; the corresponding differential identities being that of the 2nd Bianchi's (contracted). One finds the resolution of conundrum just described somewhat unsatisfactory and will return to this in a separate publication. 

\subsection{Hamiltonian analysis of the Poincar\'{e} BF theory}

Having characterized the system completely at the covariant Lagrangian level, we now pass to studying it using the canonical approach. Ordinarily the Hamiltonian methods imply the manifest breaking of covariance by explicitly separating spatial from temporal field-components, and considering equal-time Poisson brackets. We would like to stress that the chosen preferred status of time coordinate in Hamiltonian analysis is not associated \textit{a priori} with an explicit separation of spacetime \textit{itself} into ``space and time'' (not at this stage at least), as suggested e.g. by the ADM change of coordinates and their geometrical interpretation. We therefore are being cautious with usage of such notions which are usually referred to as 3+1-decomposition, or slicing/splitting/foliation/etc. In particular, all 4d symmetries persist at the canonical level in the form of gauge generators, mapping solutions into solutions, as we show for this particular example.

Following the Dirac's general treatment of singular Lagrangian systems~\cite{Dirac1964lectures}, one starts by defining the conjugate momenta (for \textit{all} configuration variables):
\begin{subequations}
\begin{align}\label{eq:null-conjugate-momenta}
&&\Pi^a_{AB} &&& := & \frac{\de L_0}{\de \dot{\om}^{AB}_a}  &&& \equiv  && \left(\Pi^0_{AB},\Pi^i_{AB}\right) && \approx && \left(0,\frac12 \eps^{ijk}B_{jkAB}^{\phantom{jkA}}\right), && \\
&&\pi^a_A  &&& := & \frac{\de L_0}{\de \dot{e}^A_a} &&& \equiv && \left(\pi^0_A,\pi^i_A\right) && \approx && \left(0,\frac12 \eps^{ijk}\be_{jkA}^{\phantom{jkA}}\right), && \\
&&\Phi^{ab}_{AB}  &&& := & \frac{\de L_0}{\de \dot{B}^{AB}_{ab}} &&& \equiv && \left(\Phi^{0i}_{AB},\Phi^{jk}_{AB}\right) && \approx && \quad 0,  && \\
&&\phi^{ab}_A  &&& := & \frac{\de L_0}{\de \dot{\be}^A_{ab}} &&& \equiv && \left(\phi^{0i}_A,\phi^{jk}_A\right) && \approx  && \quad 0, &&
\end{align}
\end{subequations}
where dot denotes the time derivative of the field variables $\dot{q} = \pa_0 q$, and $\eps^{ijk} = \eps^{0ijk},\, i,j,k=1,2,3$. We have the totality of primary constraints for the generalized coordinates $q$ and momenta $p$, in the sense that none of the velocities $\dot{q}$ enter the above relations and cannot be inverted -- the system defined by $L_0$ is, thus, maximally singular. The conjugate pairs $(\om,\Pi),(e,\pi),(B,\Phi),(\be,\phi)$ satisfy the canonical commutation relations (c.c.r.):
\begin{subequations}\label{eq:null-ccr}
\begin{align}
\big\{\om_a^{AB}(\mathbf{x}),\Pi^b_{CD}(\mathbf{y})\big\} \ & = \ \de^{[A}_{\,C}\de^{B]}_D \de_a^b \de(\mathbf{x},\mathbf{y}), & \big\{e_a^A(\mathbf{x}),\pi^b_B(\mathbf{y})\big\} \ & = \ \de^A_B \de_a^b \de(\mathbf{x},\mathbf{y}), \\
\big\{B_{ab}^{AB}(\mathbf{x}),\Phi^{cd}_{CD}(\mathbf{y})\big\} \ & = \ \de^{[A}_{\,C}\de^{B]}_D \de_{[a}^{\,c}\de_{b]}^d \de(\mathbf{x},\mathbf{y}), & \big\{\be_{ab}^A(\mathbf{x}),\phi^{cd}_B(\mathbf{y})\big\} \ & = \ \de^A_B \de_{[a}^{\,c}\de_{b]}^d \de(\mathbf{x},\mathbf{y}).
\end{align}
\end{subequations}

One then constructs the Hamiltonian, schematically $H(q,p)=p\,\dot{q}-L(q,\dot{q},p)=\phi\,\dot{q}+H_{\mathrm{c}}(q,p)$, following Dirac~\cite{Dirac1964lectures}, sometimes also called `total' in order to distinguish it from the `canonical' part $H_{\mathrm{c}}$, which does not contain primary constraints $\phi$. It is straightforward to verify that $H_{\mathrm{c}}$ is indeed explicitly independent of the velocities $\dot{q}$, which enter as undetermined functions in front of $\phi$. By this rigorous procedure one gets for the Hamiltonian density:
\begin{align}
\ms H \ = & \ \Pi^a_{AB} \dot{\om}_a^{AB} + \pi^a_A \dot{e}_a^A + \Phi^{ab}_{AB}\dot{B}^{AB}_{ab} + \phi^{ab}_A\dot{\be}^A_{ab} - \ms L \nonumber \\
 = & \ \left(\Pi^i_{AB} -\frac12 \eps^{ijk}B_{jkAB}^{\phantom{jkAB}}\right)\dot{\om}_i^{AB} + \left(\pi^i_A -\frac12 \eps^{ijk}\be_{jkA}^{\phantom{jkA}}\right)\dot{e}_i^A + \Phi^{ij}_{AB}\dot{B}^{AB}_{ij} + \phi^{ij}_A\dot{\be}^A_{ij}  \nonumber \\
& \ \ + \Pi^0_{AB} \dot{\om}_0^{AB} + \pi^0_A \dot{e}_0^A + 2\Phi^{0i}_{AB}\dot{B}^{AB}_{0i} + 2\phi^{0i}_A\dot{\be}^A_{0i} + \ms H_{\mathrm{c}}, \label{eq:null-Hamiltonian}
\end{align}
the canonical part simply consists of spatial components of Lagrangian $\mathscr{H}_c=-\mathscr{L}\big|_{\dot{\varpi}=\dot{\mathpzc{B}}=0}$. It contains no momenta whatsoever at this stage, due to primary constraints, and we retain full covariance w.r.t. Lorentz indices. The form of $\mathscr{H}_c$ will be specified shortly, after reduction is made (compare with the detailed expressions unfold in the Hamiltonian analysis of closely related BFCG theory~\cite{Mikovic-etal2016Hamiltonian-BFCG-Poincare}).

Next we calculate the development of the primary constraints $\dot{\phi}=\{\phi,H\}\equiv\chi$ in order to find out the additional consistency requirements for them to preserve in time -- the secondary constraints $\chi\approx 0$. It may happen that some combinations of constraints form a second-class (sub)system, failing to commute. This is precisely our situation, since 
\begin{equation*}
\Big\{\Phi^{ij}_{AB},\Pi^{kCD} -\frac12 \eps^{klm}B_{lm}^{CD}\Big\} = \frac12 \eps^{ijk}\de^{[C}_{\,A}\de^{D]}_B , \qquad \Big\{\phi^{ij}_A,\pi^{kB} -\frac12 \eps^{klm}\be^B_{lm}\Big\} = \frac12 \eps^{ijk}\de^B_A.
\end{equation*}
Note that the second-class nature of the initial Lagrangian $L_0$ is not the specialty of our Poincar\'{e} modification but is common to any BF theory in various spacetime dimensions. This apparent fact of the full-fledged Dirac's generalized Hamiltonian analysis is often overlooked in the canonical description of BF and related theories~\cite{CMPR2012BF+Immirzi-Hamiltonian,Buffenoir-etal2004Hamiltonian-Plebanski,MontesinosCuesta2008BF-Cartan-Ndim,Perez2013SF-review} (with rare notable exceptions, e.g.~\cite{Escalante-etal2012Hamiltonian-BF-complete,Mikovic-etal2016Hamiltonian-BFCG-Poincare}).

The presence of second-class constraints signals about the degrees of freedom which are physically non-relevant, in our case these are spatial $B$ and $\be$ components. Their velocities are the Lagrange multipliers to be determined by requiring the time preservation of the corresponding second-class set -- this allows us to express them in terms of other variables (Lagrangian equations of motion):
\begin{equation}\label{eq:Lagrangian-velocities}
\begin{aligned}
\Big\{H, \Pi^i_{AB} -\frac12 \eps^{ijk}B_{jkAB}\Big\} \ & = \ \frac12 \eps^{ijk} \Big(\dot{B}_{jkAB}^{\phantom{jkAB}}+\om_{0A}^{\phantom{0A}C}B_{jkCB}^{\phantom{jkCB}}+\om_{0B}^{\phantom{jB}C}B_{jkAC}^{\phantom{jkAC}} -e_{0[A}\be_{B]jk} \\
& \: \qquad -2\Big(\pa_j^{\phantom{I}} B_{0kAB}^{\phantom{okAB}}+\om_{jA}^{\phantom{jA}C}B_{0kCB}^{\phantom{0kCB}}+\om_{jB}^{\phantom{jA}C}B_{0kAC}^{\phantom{0kAC}}-e_{j[A}\be_{B]0k}\Big)\Big)\ = \ 0 , \\
\Big\{H, \pi^i_A -\frac12 \eps^{ijk}\be_{jkA}\Big\} \ & = \ \frac12 \eps^{ijk} \Big(\dot{\be}_{jkA}^{\phantom{jkA}}+\om_{0A}^{\phantom{0A}B}\be_{jkB}^{\phantom{jkB}}-2\Big(\pa_j^{\phantom{I}} \be_{0kA}^{\phantom{0kA}}+\om_{jA}^{\phantom{jA}B}\be_{0kB}^{\phantom{0kB}}\Big) \Big) \ = \ 0 , \\
\Big\{\Phi^{ij}_{AB},H\Big\} \ & = \ \frac12 \eps^{ijk} \Big(\dot{\om}_{kAB}^{\phantom{kAB}} +\om_{0A}^{\phantom{0A}C}\om_{kCB}^{\phantom{kCB}} -\pa_k^{\phantom{I}}\om_{0AB}^{\phantom{0AB}} -\om_{kA}^{\phantom{kA}C}\om_{0CB}^{\phantom{0CB}}\Big) \ = \ 0 , \\
\Big\{\phi^{ij}_A,H\Big\} \ & = \ \frac12 \eps^{ijk} \Big(\dot{e}_{kA}^{\phantom{kA}} +\om_{0A}^{\phantom{0A}B}e_{kB}^{\phantom{kB}} -\pa_k^{\phantom{I}} e_{0A}^{\phantom{0A}} -\om_{kA}^{\phantom{kA}B}e_{0B}^{\phantom{0B}}\Big) \ = \ 0 .
\end{aligned}
\end{equation}

One can reduce the system by solving the second-class constraints as strong equations. The formal procedure includes passing to the Dirac brackets in order not to sum over variables, which have been thrown away (cf.~\cite{Escalante-etal2012Hamiltonian-BF-complete}). In our case the constraints are of special type, such that we can make a shortcut and simply solve for the spatial $B$ and $\be$ components (together with their identically vanishing momenta), since these just serve the purpose of identifying (the spatial part of) the $\varpi$-connection's conjugate momenta. The rest of the canonical commutation relations are unaltered, as can be easily verified, and one is left in~\eqref{eq:null-Hamiltonian} with the last line $\ms H'$, where prime now signals that $B,\be$ have been solved for $\Pi,\pi$. The rest of the primary constraints are first-class and all commute among themselves. They give rise to the secondary
\begin{subequations}\label{eq:secondary-constraints}
\begin{align}
&&&&\Big\{\Pi^0_{AB},H\Big\} &&& \equiv && \chi^0_{AB} && = && \mc D_i^{\phantom{I}}\Pi^i_{AB}- e^{\phantom{A}}_{i[A}\pi^i_{B]} && \approx && 0, &&&& \\
&&&&\Big\{\pi^0_A,H\Big\}  &&& \equiv && \chi^0_A && = && \mc D_i^{\phantom{I}}\pi^i_A && \approx && 0, &&&& \\
&&&&\Big\{\Phi^{0i}_{AB},H \Big\} &&& \equiv && \chi^{0i}_{AB} &&  = && \frac12 \eps^{ijk}\left(\pa_j^{\phantom{I}} \om_{kAB}^{\phantom{kAB}} + \om_{jA}^{\phantom{jA}C}\om_{kCB}^{\phantom{kCB}} \right) && \approx && 0, &&&& \\
&&&&\Big\{\phi^{0i}_A,H \Big\} &&& \equiv && \chi^{0i}_A && = && \frac12 \eps^{ijk}\left(\mc D^{\phantom{I}}_j e_{kA}^{\phantom{kA}} \right) && \approx && 0. &&&&
\end{align}
\end{subequations}
The derivative $\mc D_i^{\phantom{I}}$ is taken w.r.t. the spatial connection $\om_i^{AB}$~\footnote{The constraint $\chi^0_{AB}$ is again the modified Gauss law, whose appearance was anticipated in~\cite{GielenOriti2010Plebanski-linear}. Here we encounter no need to artificially enlarge the phase space, which follows naturally from the covariant action, together with the nice transformation properties.}. The canonical part of the Hamiltonian takes form 
\begin{equation}\label{eq:null-canonical-Hamiltonian}
-\ms H_{\mathrm{c}}' \ = \ \om_0^{AB}\chi^0_{AB} + e_0^A\chi^0_A +2B^{AB}_{0i}\chi^{0i}_{AB} +2\be^A_{0i}\chi^{0i}_A -\pa_i^{\phantom{|}}\! \left(\om_0^{AB}\Pi^i_{AB} +e_0^A\pi^i_A\right).
\end{equation}
The bulk contribution to $\ms H'$ vanishes as the sum of (the primary as well as secondary) constraints. We did not specify any form of the boundary conditions and kept the surface term explicit. A good cross check is the consistency between the Hamiltonian $\dot{f} = \{f,H'\}$ and the Lagrangian~\eqref{eq:Lagrangian-velocities} equations of motion, once the solution to second-class constraints is taken into account. We warn the reader not to discard the primary constraints from the outset. Although the present case of reduction is very simple, in general, it may affect the symplectic structure of the original action. Moreover, the primary constraints are essential for the equivalence between the Lagrangian and Hamiltonian formulations. By keeping only the (reduced) canonical part $H_{\mathrm{c}}'$ we cannot even address the gauge transformations on the full phase-space -- only the spatial ones. 

The completion of Dirac's procedure consists in proving that $\dot{\chi}\approx 0$ are conserved. This follows from the closure of the algebra:
\begin{equation}
\begin{aligned}\label{eq:null-constraint-algebra}
\left\{\chi^0_{AB},\chi^0_{CD}\right\} \ & = \ \chi^0_{C[A}\eta_{B]D}^{\phantom{0}}-\chi^0_{D[A}\eta_{B]C}^{\phantom{0}}, \\
\left\{\chi^0_{AB},\chi^0_C\right\} \ & = \ \chi^0_{[A}\eta_{B]C}^{\phantom{0}}, \\
\left\{\chi^0_{AB},\chi^{0i}_{CD}\right\} \ & = \ \chi^{0i}_{C[A}\eta_{B]D}^{\phantom{0i}}-\chi^{0i}_{D[A}\eta_{B]C}^{\phantom{0i}}, \\
\left\{\chi^0_{AB},\chi^{0i}_C\right\} \ & = \ \chi^{0i}_{[A}\eta_{B]C}^{\phantom{0i}}, \\
\left\{\chi^0_A,\chi^{0i}_B\right\} \ & = \ -\chi^{0i}_{AB},
\end{aligned}
\end{equation}
the rest of the commutators being trivially zero. 

It is worth at this point to perform the physical degrees of freedom count, in order to make sure that the theory is indeed topological. Starting from the initial phase space of dimensionality (which we denote by putting variables in brackets) $\big[\om^{AB}_a\big]+\big[\Pi^a_{AB}\big]+\big[e^A_a\big]+\big[\pi^a_A\big]+\big[B^{AB}_{ab}\big]+\big[\Phi^{ab}_{AB}\big]+\big[\be^A_{ab}\big]+\big[\phi^{ab}_A\big]=200$, we eliminate some variables through the strong second-class equalities $\big[B^{AB}_{ij}\big]+\big[\Phi^{ij}_{AB}\big]+\big[\be^A_{ij}\big]+\big[\phi^{ij}_A\big]=60$. Finally, we perform the symplectic reduction as follows: put the system on the surface of first-class constraints, then gauge away the redundant modes by factoring out the action of the first-class constraints. In effect, we subtract twice the amount of all the first-class constraints, taking in account that some of them are reducible (namely, the secondary constraints with the vector index are not independent, but related through the spatial Bianchi's identities): $2\cdot\left(\big[\Pi^0_{AB}\big]+\big[\pi^0_A\big]+\big[\Phi^{0i}_{AB}\big]+\big[\phi^{0i}_A\big]+\big[\chi^0_{AB}\big]+\big[\chi^0_A\big]+\big[\chi^{0i}_{AB}\big]+\big[\chi^{0i}_A\big]-\big[\mc D_i^{\phantom{I}} \chi^{0i}_{AB}\big]-\big[\mc D_i^{\phantom{I}} \chi^{0i}_A+e^B_i\chi^{0i}_{BA}\big]\right)=140$. We conclude that the theory is devoid of local degrees of freedom, the only relevant ones being that of global nature, those coming from non-trivial topologies. This makes it a potential candidate for spinfoam quantization.

\paragraph{The gauge generator.} If one expects the Hamiltonian picture to represent the original theory, then it has to be shown that it correctly reproduces results of the manifestly covariant approach, in particular, the gauge symmetries, in the form of canonical transformations. The Dirac's old conjecture that all first-class constraints do generate such a transformations was formalized later by Castellani~\cite{Castellani1982Hamiltonian-generator} and others into a precise algorithm. This procedure defines the gauge generator (for arbitrary functions of time $\varepsilon(t)$)
\begin{equation}\label{eq:Castellani-generator}
\mc G(t) \ = \ \sum_{n=0}^N \sum_\al \varepsilon^{(n)}_\al \mc G^\al_{(N-n)}, \qquad  \varepsilon^{(n)}_\al =\frac{d^n}{dt^n}\varepsilon_\al^{\phantom{\al}},
\end{equation}
through the chains of first-class constraints, unambiguously constructed once the set of primary ones (first-class) $\{\al\}$ is given. The multi-index $\al$ is linked to the tensorial structure of transformations, while $(N-n)$ gives the generation number (primary/secondary/tertiary/etc.). As a by-product, knowing the derivative order of gauge transformations, one can predict the overall number $N$ of generations of constraints, and vice versa. 

The chains $\mc G^\al_{(N-n)}$ in~\eqref{eq:Castellani-generator} are constructed iteratively as follows:
\begin{equation}\label{eq:Castellani-chains}
\begin{aligned}
\mc G^\al_{(0)} \ & = \ \text{primary}, \\
\mc G^\al_{(1)} + \big\{\mc G^\al_{(0)},H\big\}\ & = \ \text{primary}, \\
& \ \vdots \\
\mc G^\al_{(N)} + \big\{\mc G^\al_{(N-1)},H\big\}\ & = \ \text{primary}, \\
\big\{\mc G^\al_{(N)},H\big\}\ & = \ \text{primary}.
\end{aligned}
\end{equation}
In the present situation the primary ones are
\begin{equation}
\mc G^0_{(0)AB} = \Pi^0_{AB} , \qquad \mc G^0_{(0)A} = \pi^0_A , \qquad \mc G^{0i}_{(0)AB} = \Phi^{0i}_{AB}, \qquad \mc G^{0i}_{(0)A} = \phi^{0i}_A,
\end{equation}
and the procedure terminates already at the secondary $n=0,1$:
\begin{equation}
\mc G^\al_{(1)}(\mathbf{x}) \ = \ -\chi^\al(\mathbf{x}) + \int d^3\mathbf{y} \, \mathpzc{A}^\al_{\phantom{\al}\be}(\mathbf{x},\mathbf{y}) \phi^\be(\mathbf{y}),
\end{equation}
where coefficient kernels $\mathpzc{A}^\al_{\phantom{\al}\be}$ are fixed by the last requirement in~\eqref{eq:Castellani-chains} to close onto the primary constraint surface (we have the identical zero due to commutation $\{\phi,\phi\}=\{\phi,\chi\}=0$). Straightforward calculation gives the total (smeared) generator
\begin{equation}\label{eq:null-generator}
\mc G(\mathpzc{U},\mathpzc{u},\Xi,\xi) \ = \ -\ms J(\mathpzc{U}) - \ms P(\mathpzc{u}) + \ms F(\Xi) + \ms T(\xi)
\end{equation}
as a combination of elementary ones:
\begin{equation}\label{eq:null-generators-elementary}
\begin{aligned}
\ms J(\mathpzc{U}) \ & = \ \int \dot{\mathpzc{U}}^{AB}\Pi^0_{AB} - \mathpzc{U}^{AB}\left(2\om^{\phantom{0A}C}_{0A}\Pi^0_{CB}-e^{\phantom{0A}}_{0A}\pi^0_B +4B^{\phantom{0iA}C}_{0iA}\Phi^{0i}_{CB}-2\be^{\phantom{0iA}}_{0iA}\phi^{0i}_B + \chi^0_{AB}\right), \\
\ms P(\mathpzc{u}) \ & = \ \int \dot{\mathpzc{u}}^A\pi^0_A - \mathpzc{u}^A\left(\om^{\phantom{0A}B}_{0A}\pi^0_B +2\be^{B}_{0i}\Phi^{0i}_{BA} + \chi^0_A\right), \\
\ms F(\Xi) \ & = \ \int \dot{\Xi}^{AB}_i\Phi^{0i}_{AB} - \Xi^{AB}_i\left(2\om^{\phantom{0A}C}_{0A}\Phi^{0i}_{CB} + \chi^{0i}_{AB}\right), \\
\ms T(\xi) \ & = \ \int \dot{\xi}^A_i\phi^{0i}_A - \xi^A_i\left(\om^{\phantom{0A}B}_{0A}\phi^{0i}_B + e^B_0\Phi^{0i}_{BA} + \chi^{0i}_A\right).
\end{aligned}
\end{equation}
This generalizes the result for the canonical gauge generator of $SO(3,1)$ BF theory, reported in~\cite{Escalante-etal2012Hamiltonian-BF-complete}. The construction provides the correct transformation properties via
\begin{equation}
\de f \ = \ \{f, \mc G \},
\end{equation}
mapping solutions into solutions (gauge symmetry). Unlike the secondary constraints~\eqref{eq:secondary-constraints}, it acts on the full phase-space of the theory:
\begin{align*}
\de\om^{AB}_0 \ & = \ - \left(\dot{\mathpzc{U}}^{AB}+\om_{0\phantom{A}C}^{\phantom{0}A}\mathpzc{U}^{CB}+\om_{0\phantom{B}C}^{\phantom{0}B}\mathpzc{U}^{AC}\right), \\ 
\de \om^{AB}_i \ & = \ - \left(\pa_i^{\phantom{I}} \mathpzc{U}^{AB}+\om_{i\phantom{A}C}^{\phantom{i}A}\mathpzc{U}^{CB}+\om_{i\phantom{B}C}^{\phantom{i}B}\mathpzc{U}^{AC}\right), \\ 
\de e^A_0 \ & = \ - \left(\dot{\mathpzc{u}}^A+\om_{0\phantom{A}B}^{\phantom{0}A}\mathpzc{u}^B\right) +\mathpzc{U}^A_{\phantom{A}B}e^B_0, \\
\de e^A_i \ & = \ - \left(\pa_i^{\phantom{I}} \mathpzc{u}^A+\om_{i\phantom{A}B}^{\phantom{i}A}\mathpzc{u}^B\right) +\mathpzc{U}^A_{\phantom{A}B}e^B_i, \\
\de \be^A_{0i} \ & = \ \mathpzc{U}^A_{\phantom{A}B}\be^B_{0i} + \left(\dot{\xi}^A+ \om_{0\phantom{A}B}^{\phantom{0}A}\xi^B_i\right), \\
\de B^{AB}_{0i} \ & = \ \mathpzc{U}^A_{\phantom{A}C}B^{CB}_{0i} + \mathpzc{U}^B_{\phantom{B}C}B^{AC}_{0i} + \left(\dot{\Xi}^{AB}_i+\om_{0\phantom{A}C}^{\phantom{0}A}\Xi^{CB}_i+\om_{0\phantom{B}C}^{\phantom{0}B} \Xi^{AC}_i\right)-e^{[A}_0\xi^{B]}_i, \\
\de \Pi^0_{AB} \ & = \ \mathpzc{U}^{\phantom{A}C}_A\Pi^0_{CB}+\mathpzc{U}^{\phantom{B}C}_B\Pi^0_{AC}- \mathpzc{u}^{\phantom{A}}_{[A}\pi^i_{B]} - \Xi_{iA}^{\phantom{iA}C}\Phi^{0i}_{CB}-\Xi_{iB}^{\phantom{iB}C}\Phi^{0i}_{AC} +\xi^{\phantom{iA}}_{i[A}\phi^{0i}_{B]}, \\
\de \Pi^i_{AB} \ & = \ \mathpzc{U}^{\phantom{A}C}_A\Pi^i_{CB}+\mathpzc{U}^{\phantom{B}C}_B\Pi^i_{AC}- \mathpzc{u}^{\phantom{A}}_{[A}\pi^i_{B]} +\eps^{ijk}\left(\pa_j^{\phantom{I}} \Xi_{kAB}^{\phantom{kAB}} + \om_{jA}^{\phantom{jA}C}\Xi_{kCB}^{\phantom{kCB}} + \om_{jB}^{\phantom{jB}C}\Xi_{kAC}^{\phantom{kAC}}+e^{\phantom{jA}}_{j[A}\xi_{B]k}^{\phantom{Bk}}\right), \\
\de\pi^0_A \ & = \ \mathpzc{U}_A^{\phantom{A}B}\pi^0_B + \xi^B_i\Phi^{0i}_{BA}, \\
\de\pi^i_A \ & = \ \mathpzc{U}_A^{\phantom{A}B}\pi^i_B + \eps^{ijk}\left(\pa_j^{\phantom{I}}\xi_{kA}^{\phantom{kA}} + \om_{jA}^{\phantom{jA}B}\xi_{kB}^{\phantom{kB}}\right), \\
\de \phi^{0i}_A \ & = \ \mathpzc{U}_A^{\phantom{A}B}\phi^{0i}_B + \mathpzc{u}^B\Phi^{0i}_{BA}, \\
\de\Phi^{0i}_{AB} \ & = \ \mathpzc{U}_A^{\phantom{A}C}\Phi^{0i}_{CB}+\mathpzc{U}_B^{\phantom{B}C}\Phi^{0i}_{AC}.
\end{align*}
The correct covariant expressions~\eqref{eq:BF-Poincare-gauge-transform} for all the Lagrangian field components (spatial as well as temporal, using also the second-class relations $\de B^{AB}_{ij}=\eps_{ijk}^{\phantom{I}}\de\Pi^{kAB},\de \be^A_{ij}=\eps_{ijk}^{\phantom{I}}\de\pi^{kA}$), are reproduced within the Hamiltonian framework, thus exhibiting the equivalence between the two pictures. 

Using the Jacobi identity, the commutator between the two consecutive transformations is given:
\begin{equation}\label{eq:gauge-algebra}
(\de_1\de_2-\de_2\de_1)f \ = \ \{\{\mc G_1,\mc G_2\}, f\}.
\end{equation}
Its elementary constituents realize the generalized matrix commutators:
\begin{align*}
\{\ms J(\mathpzc{U}_1), \ms J(\mathpzc{U}_2)\} \ & = \ -\ms J([\mathpzc{U}_1,\mathpzc{U}_2]), & [\mathpzc{U}_1,\mathpzc{U}_2]^{AB} \ & = \  \mathpzc{U}^A_{1\; C}\mathpzc{U}_2^{CB} - \mathpzc{U}^B_{1\; C}\mathpzc{U}_2^{CA} ,\\
\{\ms J(\mathpzc{U}), \ms P(\mathpzc{u})\} \ & = \ -\ms P(\mathpzc{U}\rt\mathpzc{u}), & (\mathpzc{U}\rt\mathpzc{u})^{A\phantom{B}} \ & = \ \mathpzc{U}^A_{\phantom{A}B}\mathpzc{u}^B,\\
\{\ms J(\mathpzc{U}), \ms F(\Xi)\} \ & = \ -\ms F([\mathpzc{U},\Xi]), & [\mathpzc{U},\Xi]^{AB}_i \ & = \ \mathpzc{U}^A_{\phantom{A}C}\Xi^{CB}_i-\mathpzc{U}^B_{\phantom{B}C}\Xi^{CA}_i,\\
\{\ms J(\mathpzc{U}), \ms T(\xi)\} \ & = \ -\ms T(\mathpzc{U}\rt\xi), & (\mathpzc{U}\rt\xi)^{A\phantom{B}}_i \ & = \ \mathpzc{U}^A_{\phantom{A}B}\xi^B_i, \\
\{\ms P(\mathpzc{u}), \ms T(\xi)\} \ & = \ +\ms F([\mathpzc{u},\xi]), & [\mathpzc{u},\xi]^{AB}_i \ & = \ \mathpzc{u}_{\phantom{i}}^{[A}\xi^{B]}_i.
\end{align*}

We use a chance to comment here on the relation between the dynamics and gauge in reparametrization invariant systems (cf. ``the problem of time'').
The Hamiltonian -- generator of time evolution -- in such model is a combination of first-class constraints, which are also known to generate the gauge transformations (i.e. ``unphysical'' changes in the description of the system). Working on the full phase space allows to disentangle these notions: the specific combinations of first-class constraints are different for two objects $\ms H$ and $\mc G$; the key role is played by the primary set. 

One usually defines the notion of Dirac observables w.r.t. individual constraints $\{\chi,f\}=0$ (often disregarding the primary $\phi\approx0$, working on the smaller phase space). In the quantum theory, one represents the canonical variables via operators on the appropriate Hilbert space $\mathpzc{H}$ of states of the system, the Poisson (Dirac) brackets being replaced by a commutator $[\;,\,]=i\{\;,\,\}$. For instance, in our example:
\begin{equation}
\ms J(\mathpzc{U}) \ \mapsto \ \mathpzc{U}\cdot\hat{\mc J}, \qquad \ms P(\mathpzc{u}) \ \mapsto \ \mathpzc{u}\cdot\hat{\mc P}
\end{equation}
are the elements of the (local) Poincar\'{e} algebra~\eqref{eq:Poincare-algebra}. The Dirac prescription then consists in imposing on states $\hat{\chi}_\al|\Psi\ket =0$ individually for each $\al$, which are then consistent for the first-class system.

With the distinction just pointed out between $\ms H$ and $\mc G$ on the full phase space, the function of canonical variables may satisfy two a priori distinct conditions: $\{\ms H,F\} \ = \ 0$ and/or $\{\mc G,F\} \ = \ 0$. The first can be thought of as characterizing ``evolving constants of motion'', i.e. uniquely associated with the state -- solution of the e.o.m. The state itself, however, is not uniquely defined by the Cauchy data and depends on the arbitrary functions, entering the Hamiltonian (in our reduced case, the velocities of $(e,\om,\be,B)$'s temporal components, associated with the first-class primary constraints on momenta, are not defined by the evolution equations). Thus, the functionals of the first type may depend on the gauge choice for particular Hamiltonian, whereas the second condition then characterizes ``gauge-invariant'' functionals. 

Requiring the time preservation of vanishing of currents $\mc G^\al$ on every hypersurface, by construction we have then $0=\pa_t\mc G^\al =\{\ms H,\mc G^\al\}$, and the state has to satisfy Hamiltonian e.o.m. So that symmetry generators provide an example of the first type functionals (observables). The commutation relations~\eqref{eq:gauge-algebra} express the fact that the canonical (pre-)symplectic structure is degenerate on the constraint surface. One then passes to the quotient w.r.t. the gauge directions, by considering the gauge equivalent classes of solutions as ``physical'' states. The construction of the appropriate observable algebra is of primary importance for the quantization, especially in gravitational theories, so the Dirac's ``rule of thumb'' for \emph{all} first-class constraints should be applied with certain care. 


\section{Poincar\'{e}-Pleba\'{n}ski (re)formulation of GR}\label{sec:Poincare-Plebanski}

Now, when we have the frames $e$ at our disposal among the legitimate dynamical variables, it is straightforward to implement the simplicity of the bivectors, in order to reproduce gravity sub-sector. Multiple choices of how to do this are conceivable. First of all, one can simply replace $B \ra\star e\wedge e$ directly in the action integral:
\begin{equation}\label{eq:EC+0torsion}
S[e,\om,\be] \ = \ \int_{\mc M} \frac12 \eps_{ABCD}^{\phantom{ABCD}}\,e^A\wedge e^B\wedge F^{CD}+ \be_A^{\phantom{A}}\wedge T^A.
\end{equation}
Secondly, one could try to achieve the same effect via the Lagrange multipliers approach, imposing the simplicity in its most direct sense:
\begin{equation}\label{eq:Lambda-constraints}
S[\mathpzc{B},\varpi,\La] \ = \ S_0[\mathpzc{B},\varpi] + \int \La^{AB}
\wedge \left(B_{AB}^{\phantom{AB}}-\frac12\eps_{ABCD}^{\phantom{ABCD}}\,e^C\wedge e^D\right),
\end{equation}
with the free independent multiplier 2-forms $\La$.

It turns out that the simplicity constraints can also be put into form, linear in both $B$ and $e$, which is more in the vein of current Spin Foam models. In order to stay self-contained and explicit, let us formulate the following

\begin{lemma}[linear simplicity]
Provided that the tetrad field is non-degenerate, and hence the map $e$ is invertible, the bivector field $B$ is simple if and only if either of the two equivalent sets of constraints is satisfied:
\begin{equation}\label{eq:linear-simplicity-dual}
  \left\{\begin{aligned}
  \ast B_{AB}^{ab}\, e^B_c \ &= \ 0 &&&& \forall \, c\notin\{a,b\}, \\
  \ast B_{AB}^{ab}\, e^B_b \ = \ \ast B_{AB}^{ac}\, e^B_c \ & = \ \ast B_{AB}^{ad}\, e^B_d &&&& \forall \, a\notin\{b,c,d\},
  \end{aligned}\right.
  \qquad\Leftrightarrow\qquad
  \left\{\begin{aligned}
  B_{ABab}^{\phantom{AB}}\, e^B_c \  & = \ 0  &&&& \forall \, c\in\{a,b\}, \\
  B_{ABa(b}^{\phantom{AB}}\, e^B_{c)} \  & = \ 0 &&&& \forall \, c\notin\{a,b\}.
  \end{aligned}\right.
\end{equation}
\end{lemma}

\begin{proof}
First, it is straightforward to verify that the two systems imply each other by noting that for some fixed $a\neq b$:
\begin{equation*}
\ast B_{AB}^{ab}\, e^B_{b'} \ = \ \sum_{cd} \frac12 \eps^{abcd}B_{ABcd}^{\phantom{A}} \, e^B_{b'} \ = \ \eps^{abcd}B_{ABcd}^{\phantom{A}} \, e^B_{b'} \quad \text{(no sum over $[cd]\neq[ab]$)}.
\end{equation*}
Assuming the conditions on the one side, the other one follows from here. Notice also that in our notation the dual becomes $\ast B_{AB}^{ab}=(\Si^{-1})_{AB}^{ab}$ on the constrained surface, where $\Si=e\wedge e$.

To show the necessity of these conditions for simplicity of bivectors, it is enough to cast $B=\star e\wedge e$ into form~\footnote{Having in mind the possible interpretation in terms of discrete geometry, \eqref{eq:3volume-simplicity} can be suggestively referred to as ``pyramid'', or ``3-volume'' form of simplicity constraints -- due to the nature of the object appearing on the right, and since the whole formula can be read as the expression for the volume of a pyramid with the base $\Si$ and a height $e$.}
\begin{subequations}\label{eq:pyramid-simplicity}
\begin{equation}\label{eq:3volume-simplicity}
B_{ABab}^{\phantom{AB}}\, e^B_c \ = \ \eps_{ABCD}^{\phantom{ABCD}}\, e^B_a e^C_b e^D_c \ = \ (\det e)\, e^d_A \, \eps_{dabc}^{\phantom{a}},
\end{equation}
or, equivalently:
\begin{equation}\label{eq:3volume-simplicity-dual}
\ast B_{AB}^{ab}\, e^B_c \ = \ (\det e) \, (e^a_A \de^b_c - e^b_A \de^a_c).
\end{equation}
\end{subequations}
Observing that the l.h.s. is already linear leads us to the (dual) analogue of `cross-simplicity' constraint in the first line of~\eqref{eq:linear-simplicity-dual}, when the r.h.s. in~\eqref{eq:pyramid-simplicity} is zero. In turn, the non trivial expression on the right in~\eqref{eq:pyramid-simplicity} restricts the l.h.s. in~\eqref{eq:3volume-simplicity} to be totally antisymmetric in $[abc]$, whereas it is independent of $b=c\neq a$ in~\eqref{eq:3volume-simplicity-dual}, leading to the second line of~\eqref{eq:linear-simplicity-dual}, respectively (no sum over spacetime indexes).

In order to demonstrate that the conditions~\eqref{eq:linear-simplicity-dual} are also sufficient, one follows the same reasoning as in~\cite{GielenOriti2010Plebanski-linear}. Namely, the generic bivector field can be expanded over the basis, spanned by the skew-symmetric products of $e$:
\begin{equation*}
\ast B_{AB}^{ab} \ = \ G^{ab}_{cd} \, e^c_{[A}e^d_{B]}, \qquad G^{ab}_{cd} \ = \ G^{[ab]}_{[cd]},
\end{equation*}
which after substitution into the first line of~\eqref{eq:linear-simplicity-dual} leads to
\begin{equation*}
\ast B_{AB}^{ab} \ = \ G_{ab}^{\phantom{a}} \, e^a_{[A}e^b_{B]} \qquad \text{(no sum over $ab$)}.
\end{equation*}
The individual normalization coefficients have to  satisfy symmetry $G_{ab}^{\phantom{a}} = G_{ba}^{\phantom{a}}$, $G_{aa}^{\phantom{a}} = 0$, but apart from that can be arbitrary. It is only after substitution of this ansatz into the second line of~\eqref{eq:linear-simplicity-dual} that we get the restriction
\begin{equation*}
G_{ab}^{\phantom{a}} \ = \ G_{ac}^{\phantom{a}} \ = \ G_{ad}^{\phantom{a}} \qquad \forall \, a\notin\{b,c,d\},
\end{equation*}
leading to the equality among all $G$'s. Thus, the $B$ is simple up to an overall factor, which can be eaten by appropriate normalization.
\end{proof}
  
The continuous formulation that we are advocating for is somewhat different from that of Gielen-Oriti's linear proposal~\cite{GielenOriti2010Plebanski-linear}, which uses 3-forms $\vartheta=\star e\wedge e\wedge e$, and bivectors $\Si = e\wedge e$ as independent variables, but rather represents its dual version. Before introducing the action principle, and in order to make closer contact between the two formulations, we first recall the corresponding constraint term in~\cite{GielenOriti2010Plebanski-linear} and notice that this can be rewritten as
\begin{equation}\label{eq:lin-simpl-cont}
\int d^4x \ \widetilde{\Xi}^{[ab][cde]}_A\Si^{AB}_{ab}\vartheta_{Bcde}^{\phantom{Bcde}}  \ = \ \int d^4x \ \Xi^{[ab]}_{Ac}\Si^{AB}_{ab}\tilde{\vartheta}^c_B.
\end{equation}
One can choose to work either with 3-forms~\eqref{eq:3-forms} or, equivalently, their dual densitiezed vectors:
\begin{equation}
\tilde{\vartheta}^a_A \ = \ \frac{1}{3!}\eps^{abcd}\vartheta_{Abcd}^{\phantom{A}} \ = \ (\det e^B_b)\, e^a_A.
\end{equation}
Correspondingly, the Lagrange multipliers $\Xi^{[ab]}_{Ac}$ should be exact tensors, s.t. $\widetilde{\Xi}^{[ab][cde]}_A = \frac{1}{3!}\eps^{cdef\phantom{]}}_{\phantom{Af}}\Xi^{[ab]}_{Af}$ -- tensor densities. The somewhat convoluted index symmetries that $\widetilde{\Xi}$ has to satisfy can be restated as the traceless condition on $\Xi^{[ab]}_{Ab}=0$, which upon variation then leads to the appearance of non-trivial Kronecker deltas on the right:
\begin{equation}\label{eq:Xi-variation}
\de \Xi \quad \Rightarrow \quad \tilde{\vartheta}^c_A\Si^{AB}_{ab} \ = \ \de^c_a v^B_b - \de^c_b v^B_a \qquad \text{for some} \ v^B_b.
\end{equation}
The antisymmetry in $[AB]$ and the tensorial nature of $\Si$ leave us no choice other than $v^A_a\propto \tilde{\vartheta}^A_a=e\, e^A_a$, and we get the simplicity up to an overall normalization, which is irrelevant. Applying the Hodge dual $\ast$, one restates this in terms of 3-forms, resulting from variation w.r.t. $\widetilde{\Xi}$, correspondingly:
\begin{equation}\label{eq:lin-simpl-dual}
\Si^{AB}_{ab}\vartheta_{Bcde}^{\phantom{A}} \ = \ v^A_a\eps_{bcde}^{\phantom{A}} - v^A_b \eps_{acde}^{\phantom{A}},
\end{equation}
which are essentially the original Gielen-Oriti's constraints.

Comparing~\eqref{eq:Xi-variation} with~\eqref{eq:3volume-simplicity-dual}, and juxtaposing them against the constraint term~\eqref{eq:lin-simpl-cont}, then suggests the respective least action principle in terms of dual variables $B\leftrightarrow\Si$ and $e\leftrightarrow\vartheta$, correspondingly:
\begin{equation}\label{eq:Poincare-Plebanski}
S_{\mathrm{PP}}[\mathpzc{B},\varpi,\Theta] \ := \ S_0[\mathpzc{B},\varpi] + \int 
\left(\Theta^A\lrcorner \, e^B\right)\wedge B_{AB}^{\phantom{AB}},
\end{equation}
which we coined, referring to its gauge group, the Poincar\'{e}-Pleba\'{n}ski formulation (although such a name might be as well attributed either to the ``$\La$-version'', or essentially to any formulation of this flavour). The $4\times4\times6=96$ Lagrange multipliers constitute the tangent $T\mc M$-valued 2-forms, that is
\begin{equation}\label{eq:PP-Lagrange-multipliers}
\Theta^A \ = \ \frac12 \Theta^{Ac}_{ab}\, dx^a\wedge dx^b \otimes \pa_c^{\phantom{a}}
\end{equation}
are the sections of the fiber bundle $\bigwedge^2 T^\ast\mc M\bigotimes T\mc M$. In the constraint term of the action~\eqref{eq:Poincare-Plebanski} they contract with tetrad 1-forms using the pairing $dx^a_{\phantom{b}}\lrcorner\,\pa_b^{\phantom{a}}=\de^a_b$ in the tangent vector index: $\Theta^A\lrcorner \, e^B=\frac12 \Theta^{Ac}_{ab}e^B_c dx^a\wedge dx^b$. The $\Theta$'s are restricted to be traceless $\Theta^{Ab}_{[ab]}=0$, that is possess the components of the form $\Theta^{Ac}_{ab}+\frac23\de^{\,c}_{[a}\Theta^{Ad}_{b]d}$; we can formulate this in the coordinate independent way as the full contraction with the canonical tangent-valued form on $\mc M$ being zero:
\begin{subequations}
\begin{equation}\label{eq:traceless-multipliers}
\theta_{\mc M} \ := \ dx^a\otimes \pa_a, \qquad \theta_{\mc M}\lrcorner\, \Theta^A \ = \ 0.
\end{equation}

Lets count the number of independent $\Theta$ components, in order to verify that we have enough of them to eliminate $36$ $B^{CD}_{cd}$ in favour of $16$ $e^A_a$. Apart from the $4\times4$ traceless conditions~\eqref{eq:traceless-multipliers}, from the contraction with $B$ in the action~\eqref{eq:Poincare-Plebanski} follow $10\times6$ antisymmetrization equations
\begin{equation}\label{eq:Lagrange-multipliers-symmetries}
\Theta^{(A}\lrcorner \, e^{B)} \ = \ 0,
\end{equation}
\end{subequations}
which $\Theta$ and $e$ have to satisfy. Subtracting from this $16$ d.o.f. corresponding to $e$'s (they just serve the purpose to isomorphically map indices $e(x):T_x\mc M\rightarrow \mathbb{R}^{3,1}$), we are left with $60-16=44$ additional requirements on $\Theta$. Thereby we get the total number of $96-16-44=36$ independent $\Theta$'s -- exactly the right amount to enforce simplicity.

It should be more or less evident after our exposition that the symmetries of the Lagrange multipliers lead to the variation, constrained by the system~\eqref{eq:linear-simplicity-dual}, depending whether we choose to vary w.r.t. $\Theta$ or its dualized version $\widetilde{\Theta}$. The lemma then implies that this is the same as performing variation on the simplicity constraint surface. The manifest presence of the Hodge-star $\star$ in constraint~\eqref{eq:Lambda-constraints} becomes shrouded, instead one has the restriction on the multipliers~$\Theta$. The free variation of $\de\La$ equates the constraint pre-factor to zero exactly, whilst for~$\de\Theta$ obeying additional conditions -- we get the non-vanishing expression on the r.h.s. In an analogous situation within the  standard Pleba\'{n}ski quadratic approach, the corresponding quantity on the right is usually interpreted in geometric terms as a definition of the 4-volume (on the solution of constraints), whilst a non-trivial symmetrization conditions are put on the l.h.s. It is these latter conditions that actually constitute the substance of the respective `volume' part of simplicity constraints. They require that the definition of the 4-volume be consistent, i.e. does not depend on the multiple choices that could be made for its parametrization on the l.h.s.
Note that in~\eqref{eq:pyramid-simplicity} we get the very same picture, now with the quantity on the r.h.s. being precisely the non-trivial \emph{3-volume} (cf~\eqref{eq:volume-unique}). At the same time this last bit now is \emph{``localized''} at the level of each tetrahedron, irregardless of the whole 4-simplex, which was the case for quadratic version~\eqref{volume-quadr}. Lastly, the analogue of the `cross-simplicity', when the r.h.s. is zero, now expresses that the corresponding (discrete) $e$ is collinear with the face $S_f$, being orthogonal to its dual bivector $B_f=\star\Si_f$.

These 3 \textit{a priori} distinct choices for constraint imposition, tabulated above, all seem to represent the same physical content. In either of the $\La$ or $\Theta$ versions, inserting further the solution for $B$ back into action, one reduces the initial topological theory~\eqref{eq:BF-Poincare} to that of~\eqref{eq:EC+0torsion}, that is the Einstein-Cartan action~\eqref{eq:Einstein-Cartan} supplemented with an extra term for (zero) torsion. At first sight, this might seem an excess, since the variations $\de\om$ of the EC-term alone incidentally give the vanishing of $De=0$ on-shell. However, the two theories are not identical: $\be$ plays the role of the Lagrange multiplier imposing the torsionless constraint $T=0$, which by the rule of procedure is required when we pass from the 2nd to 1st order formulation, in order to preserve the original dynamical content of the theory. Thus, we expect the equivalence should hold with the Einstein-Hilbert variational principle, and not with the ``Palatini variation'' method. The relations between different action principles can be schematically depicted in a diagram:
\begin{displaymath}
\begin{tikzcd}[column sep = large, row sep=large]
\text{1st order:} & S_{\mathrm{PP}}[\mathpzc{B},\varpi,\Theta] \arrow{r}{\de\Theta} \arrow{d}[swap]{\de \be}
& S_{\mathrm{EC}}[e,\om] +\int \be\wedge T \arrow{d}{\de \be}\\
\text{2nd order:} & S_{\mathrm{GO}}[B,e,\om[e],\Theta] \arrow{r}{\de\Theta} & S_{\mathrm{EC}}[e,\om[e]] \equiv S_{\mathrm{EH}}[e]
\end{tikzcd}
\end{displaymath}
In the bottom left corner appears a variant of the ``hybrid'' action of the form dual to that of Gielen-Oriti~\cite{GielenOriti2010Plebanski-linear}, but with the unique $e$-compatible torsion-free spin connection. This is to be contrasted with their 1st order formulation, where $\om$ is independent and the gauge status of non-dynamical $\vartheta\sim e$ is less clear, which enters a separate sequence:
\begin{displaymath}
\begin{tikzcd}
S_{\mathrm{PP}}[\mathpzc{B},\varpi,\Theta]+\int \mu\wedge \be \arrow{r}{\de\mu} & S_{\mathrm{GO}}[B,e,\om,\Theta] \arrow{r}{\de\Theta} & S_{\mathrm{EC}}[e,\om] .
\end{tikzcd}
\end{displaymath}

We stress that the reduction of the Einstein-Cartan theory to that of GR is achieved only \textit{on-shell} in vacuum, by solving the dynamical e.o.m. for $\om$. In contrast, one puts additional restrictions on the allowed variations of the generalized coordinates by the use of (non-dynamical) Lagrange multipliers $\Theta,\be$, which then acquire the physical meaning of ``reaction forces'', corresponding to variations that violate the constraints. The discussions of the relation between two approaches have been recurrent in the literature in the past, in particular, regarding the higher order Lagrangians and matter couplings (e.g., see~\cite{Kichenassamy1986Lagrangian-multipliers-grav} and references therein). It is a firmly established fact that the constrained variations should lead to the same result as for the case where constraints have been already solved from the outset. On the other hand, by allowing arbitrary variations w.r.t. d.o.f. which were previously restrained to lie on the constrained surface~\footnote{Often referred to as ``Palatini variation'' -- though somewhat erroneously (see discussion in~\cite{Ferraris-etal1982Palatini-history}).} the presumed ``equivalence'' with the original (2nd order) set-up is skewed. One plainly does not possess the same d.o.f. in two approaches to variation. This might be of relevance for the precise form of dynamical symmetries, by the Noether's theorem, since the variations would in general contain terms off the constraint surface.

\section{Conclusions and outlook}\label{sec:summary}

In the first part of this work we reviewed the classical Pleba\'{n}ski formulation of gravity, which underlies current Spin Foam models of EPRL and FK. We considered both the quadratic version and the fully linear formulation with 3-forms (4d normals), paying a special attention to the implementation of the `volume' part of simplicity constraints. Our revision of its quadratic version in the symmetry reduced setting of cuboids revealed that one cannot replace it with the 3d closure condition, in general, contrary to triangulations. As result, there is no unique geometric 4-volume. The linear case puts non-trivial conditions on normals and bivectors, which ensure the existence of edges/matching of shapes/uniquely defined 3-volumes, in quite an intricate way. This prompts us to pass to the tetrad variables, instead of normals. 

In the second part of this work we considered the modification of the classical action principle by putting torsion to zero. To explore the consequences we studied in Sec.~\ref{sec:Poincare-BF} the corresponding change in the BF theory and its larger gauge group -- the Poincar\'{e} (affine) extension of the (homogeneous) Lorentz group -- developing our analysis in detail both at the Lagrangian and Hamiltonian levels. Our Dirac's generalized constraint analysis, in fact, shortcuts the derivation in~\cite{Mikovic-etal2016Hamiltonian-BFCG-Poincare}, corroborating also the equivalence with covariant framework by an actual construction of canonical gauge generator. The extended action presents a perfect ground for the imposition of simplicity constraints, as we have frames explicitly at our disposal. In Sec.~\ref{sec:Poincare-Plebanski} we present an alternative look at the linear simplicity constraints, dual to that of~\cite{GielenOriti2010Plebanski-linear}, and comment on the relations between various formulations.

One point to be noticed is that in our reformulation it becomes apparent that the tetrad field by and large presents in the EPRL model right away. It sneaks in disguise of a 4-vector normal to the tetrahedra, required for linear reformulation of constraints. Recalling the original Pl\'{e}banski's quadratic constraints, this was based on the appearance of tetrad in the action only in a specific combination, which could be collectively \emph{denoted} $B=\star e\wedge e$, and has to satisfy some algebraic relations. The same situation is encountered in LQG with the conjugate momenta and reality conditions~\cite{Alexandrov2006Reality-from-CovarintLQG}. However, let us stress that the role of soldering form as an independent entity is much more than that: it encodes all the metric properties and is directly related to diffeomorphisms, playing the part of the gauge potential of local translations. These roles of tetrad are hardly appreciated in a formulation, where it is basically excluded from consideration as a configuration variable, and diminished to just the momenta. In essence, upon a closer look the theory is indistinguishable from the classical Einstein-Cartan, or Einstein-Hilbert gravity (leaving aside the Barbero-Immirzi parameter), and one has to face the task of directly handling tetrads at the discrete and quantum levels.~\footnote{The similar views were recently advocated in~\cite{Charles2017simplicity-3d}, within the context of canonical 3d LQG.}

Despite the classical nature of our results, it is clear from the context in which we put the present work, that we expect the view developed here to be of relevance for Quantum Gravity, both at the level of path integral (SF) and for covariant canonical loop-quantization. In particular, the Pleba\'{n}ski approach applied to the BF-Poincar\'{e} theory displays minimal distinction from the correspondingly constrained 1st order formulation, thus reinforcing one's expectations for a better contact between the two quantization programs. As an outlook, let us observe several issues that one may face on this route.
\begin{itemize}
\item The choice of an appropriate discretization for a frame field and constraints. The guidance may be provided by a twofold nature of the frame field: 
\begin{itemize}
\item
On the one hand side, the vector $V=\mathbbm{R}^{3,1}$-valued 1-form $e$, being a gauge potential of translations, combines naturally with the homogeneous Lorentz part $\om\in\mathfrak{h}$ into a single Cartan connection 1-form $\varpi=\om+e$, taking values in larger algebra $\mathfrak{g}$. Treating both parts on the same footing -- at least in the topological BF case -- thus suggests a corresponding discretization in terms of generalized $G$-holonomies: besides the parallel transport along paths in $\mc M$ (given by the Ehresmann $H$-connection), a Cartan connection gives also a notion of `development on the model $G/H$-Klein geometry' (the ``rolling'', or translation of the point of tangency; see~\cite{Wise2010Cartan-geometry} and references therein)~\footnote{We assume, an incorporation of this notion might pave the way for a better control over diffeos in the discrete and quantum gravity.}.
\item
On the other side, in the constrained case of gravity $e$ provides the basis for the geometric `simple' $n$-forms (and for the $\om$'s conjugate momenta 2-forms, in this way). In the end, one could expect similarity of the discretization with the variant of Regge calculus that comes from the gauge theoretic approach to gravity~\cite{Caselle-etal1989Poincare-calculus,Gionti2005discrete-Poincare-gravity}. The discrete $e$-field is likely to appear in the `integrated' form (i.e. conjugated by the $H$-holonomies) in order to ensure gauge invariance. The dual form~\eqref{eq:3volume-simplicity} of simplicity constraints associates naturally a 3-volume normal vector to every boundary polyhedron, so the generalization to a higher valence case seems to be within the reach (e.g. using some variant of Minkowski's theorem).
\end{itemize}
\item
It is sensible to first gain some experience with the SF/loop quantization of the respective topological BF theory. The $H$-holonomy encodes the information about connection up to transformations, leaving source and target intact. Note that one gets an element of $G\supset H$ for the lift of a curve w.r.t. Cartan connection $\mathfrak{g}\cong\mathfrak{h}\oplus V$, where the latter isomorphism should hold as $\mathrm{ad}_H$-representations, for reductive Cartan geometries. The space of connections is now significantly larger, so that $H$-transformations cannot reduce the gauge d.o.f. sufficiently.

Allowing for more gauge transformations could deal with this issue, and indeed -- the full $\mathrm{ad}_G$ symmetry~\eqref{eq:Poincare-gauge-transform} suggests to use the basis of $G$-invariant spin networks in the topological BF case. Note that $\varpi$ is essentially Ehresmann connection in the $G$-bundle, corresponding to vertical automorphisms. However, the reduction to gravity attaches the affine frames to the manifold, while soldering the $\mathbb{R}^{3,1}$ part to the tangent directions, corresponding to (horizontal) diffeomorphisms. The connection here is Cartan's absolute parallelism (on $H$-bundle). This distinction seems to be crucial, and the usual procedure should be specialized correspondingly. 


\end{itemize}

\section*{Acknowledgments}
The author is grateful to Benjamin Bahr and Sebastian Steinhaus for stimulating discussions and for useful comments on an earlier draft of this paper. This work was funded by the project BA 4966/1-1 of the German Research Foundation (DFG).

\appendix
\renewcommand{\theequation}{\Alph{subsection}.\arabic{equation}}
\section*{Appendices}
\addcontentsline{toc}{section}{Appendices}
\renewcommand{\thesubsection}{\Alph{subsection}}

\subsection{Conventions and notation}
\label{app1}

We use the Latin letters from the beginning of the alphabet to denote the covariant field components: lowercase $a,b,c,...=0,1,2,3$ for the world tensors w.r.t. the holonomic coordinate basis, and capital $A,B,C,...=0,1,2,3$ for the (internal) Lorentz coordinates w.r.t. the orthonormal locally inertial frames. For the $3+1$ space/time split, the letters from the middle of the alphabet are used $i,j,k,...=1,2,3$ for the spatial field components.

Anti-symmetrization of indices is performed with the respective order $|S_n|=n!$ of the symmetry group in the denominator, s.t the projection property holds, and denoted by the square brackets:
\begin{equation}
t_{[a_1...a_n]} \ := \ \frac{1}{n!}\sum_{\pi\in S_n}\mathrm{sign}(\pi) \, t_{a_{\pi(1)}...a_{\pi(n)}}.
\end{equation}

The internal $\mathbb{R}^{3,1}$ comes with the Minkowski metric $\eta_{AB}=\mathrm{diad}(-1,1,1,1)$ and the totally anti-symmetric Levi-Civita symbol $\eps^{ABCD}_{\phantom{ABCD}}=4! \,\de^{[A}_{\, 0}\de^{B\phantom{|}}_{\, 1}\de^{C\phantom{|}}_{\, 2}\de^{D]}_{\, 3}$. (Note that, although we prefer ``mostly plus'' convention for the metric signature, we maintain the full Lorentz covariance and actually never use this explicitly.) The metric $\eta$ allows to freely raise and lower internal indices, identifying $\mathbb{R}^{3,1}$ with its dual, e.g. $\eps_{ABCD}^{\phantom{ABC|}}=\eta_{AA'}^{\phantom{A|}}\eta_{BB'}^{\phantom{B|}}\eta_{CC'}^{\phantom{C|}}\eta_{DD'}^{\phantom{D|}}\eps^{A'B'C'D'} =-4!\, \de_{[A}^{\, 0}\de_{B\phantom{|}}^{\, 1}\de_{C\phantom{|}}^{\, 2}\de_{D]}^{\, 3}$. It defines the (internal) Hodge-star duality operator (in arbitrary dimension $N$):
\begin{equation}
\begin{aligned}
&\star \ : & & \bigwedge^{N-n} \mathbb{R}^{N-1,1} && \ra && \bigwedge^{n} \mathbb{R}^{N-1,1} \ \quad \forall n,  && \\ 
&& & Q^{A_1...A_{N-n}} && \mapsto && \star Q_{A_1...A_n}^{\phantom{AD}} \ := \ \frac{1}{n!}\eps_{A_1...A_nB_{n+1}...B_N}^{\phantom{AD}}Q^{B_{n+1}...B_N}. &&
\end{aligned}
\end{equation}

Analogously, the spacetime totally-skew symbol $\eps^{abcd}_{\phantom{abc\phantom{|}}}$ is defined to have $\eps^{0123}=\eps_{0123}^{\phantom{012}}=1$ for both upper and lower indices. To it corresponds the respective Hodge-star $\ast$ on spacetime exterior algebra. The contraction properties are
\begin{equation}
\eps^{a_1...a_nc_{n+1}...c_N}\eps_{b_1...b_nc_{n+1}...c_N}^{\phantom{abc|}} \ = \ n!(N-n)! \,\de^{[a_1}_{\, b_1}\cdots\de^{a_n]}_{\, b_n},
\end{equation}
and similarly for the internal $\eps$, just with the minus sign, in particular $\eps^{ABCD}\eps_{ABCD}=-4!$. Working with $\eps$ allows to coveniently express the determinants, e.g. the elementary 4d volume spanned by the coordinate basis (co)vectors 
\begin{equation}\label{app:d^4x}
dx^a\wedge dx^b\wedge dx^c\wedge dx^d  \ = \ d^4x \, \eps^{abcd}.
\end{equation}
Intuitively, $\eps$ (in conjunction with the frame $e$) generalizes the flat Euclidean vector cross-product to arbitrary spacetimes and their subspaces. Thus, we get the dual vectors~\eqref{eq:3-forms}, representing locally the 3d volume normal to the elementary parallelepiped (converted to the orthogonal cuboid in the locally inertial frame $e$) of the hypersurface. Similarly, the simple bivectors~\eqref{eq:Plebanski-sectors} from the $II$-sector represent locally the 2d area normals to the surface's elementary parallelograms. 

We use the $\lrcorner\,$-symbol for contractions of tensors' components, employing the duality of elementary coordinate co/vectors $dx^a_{\phantom{b}}\lrcorner\,\pa_b^{\phantom{a}}=\de^a_b$. For example, the internal product of a vector with coordinate basis $n$-forms:
\begin{equation}\label{app:interior-product}
\pa_b^{\phantom{a}}\lrcorner \left( dx^{a_1}_{\phantom{b}}\wedge ... \wedge dx^{a_n}_{\phantom{b}}\right) \ = \ n! \, \de^{[a_1}_{\,b} dx^{a_2\phantom{\!\!]}}_{\phantom{b}}\wedge ... \wedge dx^{a_n]}_{\phantom{b}}.
\end{equation}
It is ordinarily clear from the context which of the components of tensors are being contracted.

\subsection{On the teleparallel ``gauge''}
\label{app3}
As a side remark, let us touch upon how one can alternatively arrive, starting from the same unconstrained action~\eqref{eq:BF-Poincare}, to the so called `teleparallel equivalent of GR'. In place of simplicity for $B$, one may choose to constrain~$\be$ in the original Poincar\'{e} BF action~\eqref{eq:BF-Poincare}. One can split up the generic metric-preserving connection into $e$-compatible (torsionless, Levi-Civita) part and contortion tensor $K$, respectively:
\begin{equation}
\om^{AB} \ := \ \om[e]^{AB} + K^{AB}, \qquad T^A \ = \ K^A_{\phantom{A}B}\wedge e^B.
\end{equation}
If we require $\be$ to be of the form:
\begin{equation}
\be_A^{\phantom{A}} \ = \ \frac12 \eps_{ABCD}^{\phantom{ABCD}}e^B\wedge K^{CD},
\end{equation}
we obtain the theory of distant parallelism with non-trivial torsion, written in components as follows:
\begin{equation}
S_{\|} \ = \ \int d^4x \, e\left( \la^{ab}_{AB} F^{AB}_{ab}-\frac14 T_{ab}^{\phantom{ab}c}T^{ab}_{\phantom{ab}c}+\frac12 T_{ab}^{\phantom{ab}c}T_c^{\phantom{c}ab}+T_{ac}^{\phantom{ac}a}T^{bc}_{\phantom{bc}b}\right), \qquad e \, \la^{ab}_{AB} \ := \ \frac12 \ast B^{ab}_{AB}.
\end{equation}
The relation to General Relativity is established via following identity:
\begin{equation}\label{eq:EC-teleparallel}
e\left( e^a_Ae^b_B F^{AB}_{ab}-\frac14 T_{ab}^{\phantom{ab}c}T^{ab}_{\phantom{ab}c}+\frac12 T_{ab}^{\phantom{ab}c}T_c^{\phantom{c}ab}+T_{ac}^{\phantom{ac}a}T^{bc}_{\phantom{bc}b}\right) \ = \ e R- \pa_a\left(e \, e^a_Ae^b_BK_b^{\phantom{B}AB}\right),
\end{equation}
where $R$ is the Ricci scalar curvature (built from Levi-Civita connection). That is for the vanishing curvature $F$ (of the so called Weitzenb\"{o}ck connection), the part quadratic in torsion differs from the Einstein's theory by a total divergence. The identity~\eqref{eq:EC-teleparallel} also demonstrates explicitly the difference between EC and EH Lagrangians.


\bibliography{LQG-SpinFoams-bibl}

\end{document}